\documentclass[a4paper,UKenglish,cleveref, autoref]{lipics-v2019}

\nolinenumbers

\usepackage{tikz}
\usetikzlibrary{automata, positioning, arrows}

\newcommand{\pO}{{\sf Player 1}}
\newcommand{\pT}{{\sf Player 2}}
\newcommand{\AO}{A}
\newcommand{\AT}{B}
\newcommand{\aO}{a}
\newcommand{\aT}{b}
\newcommand{\Col}{C}

\newcommand{\tr}{\mathrm{col}}

\newcommand{\hOS}{h}
\newcommand{\ol}{\overline}

\newcommand{\aH}{\rho}

\newcommand{\aHT}{\beta}
\newcommand{\aHinf}{\boldsymbol{\aH}}

\newcommand{\aHTinf}{\boldsymbol{\aHT}}

\newcommand{\sO}{s}

\newcommand{\mf}{\mu}

\newcommand{\R}{\mathbb{R}}
\newcommand{\N}{\mathbb{N}}
\newcommand{\e}{\epsilon}
\newcommand{\lex}{\mathrm{lex}}

\title{Time-aware uniformization of winning strategies} 

\titlerunning{Time-aware uniformization of winning strategies}

\author{St{\'e}phane Le Roux}{LSV, Universit{\'e} Paris-Saclay, ENS Paris-Saclay, CNRS, France \and \url{http://www.lsv.fr/~leroux/} }{leroux@lsv.fr}{}{}

\authorrunning{S. Le Roux}

\Copyright{St{\'e}phane Le Roux}

\ccsdesc[100]{General and reference~General literature}
\ccsdesc[100]{General and reference}

\keywords{Two-player win/lose games, imperfect information, criterion for existence of uniform winning strategies, finite memory}

\category{}

\relatedversion{}

\supplement{}

\acknowledgements{This article benefited from a useful conversation with Fran\c{c}ois Schwarzentruber and St{\'e}phane Demri, and from very careful reading and comments by an anonymous referee.}

\EventEditors{John Q. Open and Joan R. Access}
\EventNoEds{2}
\EventLongTitle{42nd Conference on Very Important Topics (CVIT 2016)}
\EventShortTitle{CVIT 2016}
\EventAcronym{CVIT}
\EventYear{2016}
\EventDate{December 24--27, 2016}
\EventLocation{Little Whinging, United Kingdom}
\EventLogo{}
\SeriesVolume{42}
\ArticleNo{23}

\begin{document}

\maketitle

\begin{abstract}
Two-player win/lose games of infinite duration are involved in several disciplines including computer science and logic. If such a game has deterministic winning strategies, one may ask how simple such strategies can get. The answer may help with actual implementation, or to win despite imperfect information, or to conceal sensitive information especially if the game is repeated.

Given a concurrent two-player win/lose game of infinite duration, this article considers equivalence relations over histories of played actions. A classical restriction used here is that equivalent histories have equal length, hence \emph{time awareness}. A sufficient condition is given such that if a player has winning strategies, she has one that prescribes the same action at equivalent histories, hence \emph{uniformization}. The proof is fairly constructive and preserves finiteness of strategy memory, and counterexamples show relative tightness of the result. Several corollaries follow for games with states and colors.
\end{abstract}

\section{Introduction}

In this article, two-player win/lose games of infinite duration are games where two players \emph{concurrently and deterministically} choose one action each at every of infinitely many rounds, and ``in the end'' exactly one player wins. Such games (especially their simpler, turn-based variant) have been used in various fields ranging from social sciences to computer science and logic, e.g. in automata theory \cite{EJ91, Mostowski91} and in descriptive set theory \cite{Martin75}.

Given such a game and a player, a fundamental question is whether she has a winning strategy, i.e. a way to win regardless of her opponent's actions. If the answer is positive, a second fundamental question is whether she has a simple winning strategy. More specifically, this article investigates the following \emph{strategy uniformization} problem: consider an equivalence relation $\sim$ over histories, i.e. over sequences of played actions; if a player has a winning strategy, has she a winning $\sim$-strategy, i.e. a strategy prescribing the same action after equivalent histories? This problem is relevant to imperfect information games and beyond.

This article provides a sufficient condition on $\sim$ and on the winning condition of a player such that, if she has a winning strategy, she has a winning $\sim$-strategy. The sufficient condition involves time awareness of the player, but \emph{perfect recall} (rephrased in Section~\ref{sect:main-def-res}) is not needed. On the one hand, examples show the tightness of the sufficient condition in several directions; on the other hand, further examples show that the sufficient condition is not strictly necessary.

The proof of the sufficient condition has several features. First, from any winning strategy $\sO$, it derives a winning $\sim$-strategy $\sO \circ f$. The map $f$ takes as input the true history of actions, and outputs a well-chosen \emph{virtual history} of equal length. Second, the derivation $\sO \mapsto \sO \circ f$ is $1$-Lipschitz continuous, i.e., \emph{reactive}, as in reactive systems. (Not only the way of playing is reactive, but also the synthesis of the $\sim$-strategy.) Third, computability of $\sim$ and finiteness of the opponent action set make the derivation \emph{computable}. As a consequence in this restricted context, if the input strategy is computable, so is the uniformized output strategy. Fourth, \emph{finite-memory} implementability of the strategies is preserved, if the opponent action set is finite. Fifth, strengthening the sufficient condition by assuming \emph{perfect recall} makes the virtual-history map $f$ definable incrementally (i.e. by mere extension) as the history grows. This simplifies the proofs and improves the memory bounds.

The weaker sufficient condition, i.e. when not assuming perfect recall, has an important corollary about concurrent games with states and colors: if \emph{any} winning condition (e.g. not necessarily Borel) is defined via the colors, a player who can win can do so by only taking the history of colors and the current state into account, instead of the full history of actions. Finiteness of the memory is also preserved, if the opponent action set is finite. Two additional corollaries involve the energy winning condition or a class of winning conditions laying between B{\"u}chi and Muller.

Both the weaker and the stronger sufficient conditions behave rather well algebraically. In particular, they are closed under arbitrary intersections. This yields a corollary involving the conjunction of the two aforementioned winning conditions, i.e., energy and (sub)Muller.

Finding sufficient conditions for strategy uniformization may help reduce the winning strategy search space; or help simplify the notion of strategy: instead of expecting a precise history as an input, it may just expect an equivalence class, e.g. expressed as a simpler trace.

The strategy uniformization problem is also relevant to \emph{protagonist-imperfect-information} games, where the protagonist cannot distinguish between equivalent histories; and also to \emph{antagonist-imperfect-information} games, where the protagonist wants to behave the same after as many histories as possible to conceal information from her opponent or anyone (partially) observing her actions: indeed the opponent, though losing, could try to lose in as many ways as possible over repeated plays of the game, to learn the full strategy of \pO\/, i.e. her capabilities. In connection with the latter, this article studies the strategy \emph{maximal} uniformization problem: if there is a winning strategy, is there a maximal $\sim$ such that there is a winning $\sim$-strategy? A basic result is proved and examples show its relative tightness.

\paragraph*{Related works}

The distinction between perfect and imperfect information was already studied in \cite{NM44} for finite games. Related concepts were clarified in \cite{Kuhn50} by using terms such as information partition and perfect recall: this article is meant for a slightly more general setting and thus may use different terminologies.

I am not aware of results similar to my \emph{sufficient condition for universal existence}, but there is an extensive literature, starting around \cite{Reif84}, that studies related \emph{decision problems of existence}: in some class of games, is the existence of a uniform winning strategy decidable and how quickly? Some classes of games come from strategy logic, introduced in \cite{CHP10} and connected to information imperfectness, e.g., in \cite{BMMRV17}. Some other classes come from dynamic epistemic logic, introduced in \cite{Hintikka62} and connected to games, e.g., in \cite{Benthem01} and to decision procedures, e.g., in \cite{MPS19}. Among these works, some~\cite{MP14} have expressed the need for general frameworks and results about uniform strategies; others~\cite{BD18} have studied subtle differences between types of information imperfectness.

Imperfect information games have been also widely used in the field of security, see e.g. the survey \cite[Section 4.2]{RESDSW10}. The aforementioned strategy maximal uniformization problem could be especially relevant in this context.





\paragraph*{Structure of the article}

Section~\ref{sect:main-def-res} presents the main results on the strategy uniformization problem; Section~\ref{sect:def-cor-sc} presents various corollaries about games with states and colors; Section~\ref{sect:sasc} proves the main result and a corollary under the stronger assumption that yields a stronger conclusion; Sections~\ref{sect:wawc} and \ref{sect:memory} prove the main result and the remaining corollaries under the weaker assumption that yields a weaker conclusion, respectively in the general case and in the memory-aware case; Section~\ref{sect:cbusc} discusses notions and properties used in the main results as well as alternative notions; Section~\ref{sect:tight} shows the tightness of the sufficient condition in several directions; finally, Section~\ref{sect:mhc} presents basic results for the strategy \emph{maximal} uniformization problem and examples showing some tightness.

\section{Main definitions and results}\label{sect:main-def-res}

The end of this section discusses many aspects of the forthcoming definitions and results.

\paragraph*{Definitions on game theory}

In this article, a {\bf two-player win/lose game} is a tuple $\langle \AO, \AT, W\rangle$ where $\AO$ and $\AT$ are non-empty sets and $W$ is a subset of infinite sequences over $\AO \times \AT$, i.e. $W \subseteq (\AO \times \AT)^\omega$. Informally, \pO\/ and \pT\/ concurrently choose one action in $\AO$ and $\AT$, respectively, and repeat this $\omega$ times. If the produced sequence is in $W$, \pO\/ wins and \pT\/ loses, otherwise \pT\/ wins and \pO\/ loses. So $W$ is called the winning condition (of \pO\/).

The {\bf histories} are the finite sequences over $\AO \times \AT$, denoted by $(\AO \times \AT)^*$. The {\bf opponent-histories} are $\AT^*$. The {\bf runs} and {\bf opponent-runs} are their infinite versions.

A \pO\/ {\bf strategy} is a function from $\AT^*$ to $\AO$. Informally, it tells \pO\/ which action to choose depending on the opponent-histories, i.e. on how \pT\/ has played so far.

The {\bf induced history} function $\hOS: ((\AT^* \to \AO) \times \AT^*) \to (\AO \times \AT)^*$ expects a \pO\/ strategy and an opponent-history as inputs, and outputs a history. It is defined inductively: $\hOS(\sO,\e) := \e$ and $\hOS(\sO,\aHT\cdot \aT) := \hOS(\sO,\aHT) \cdot (\sO(\aHT),\aT)$ for all $(\aHT,\aT) \in \AT^* \times \AT$. Informally, $h$ outputs the very sequence of pairs of actions that are chosen if \pO\/ follows the given strategy while \pT\/ plays the given opponent-history. Note that $\aHT \mapsto \hOS(\sO,\aHT)$ preserves the length and the prefix relation, i.e. $\forall \aHT,\aHT' \in \AT^*, |\hOS(\sO,\aHT)| = |\aHT| \wedge (\aHT \sqsubseteq \aHT' \Rightarrow \hOS(\sO,\aHT) \sqsubseteq  \hOS(\sO,\aHT'))$.

The function $\hOS$ is extended to accept opponent-runs (in $\AT^\omega$) and then to output runs: $\hOS(\sO,\aHTinf)$ is the only run whose prefixes are the $\hOS(\sO,\aHTinf_{\leq n})$ for $n \in \N$, where $\aHTinf_{\leq n}$ is the prefix of $\aHTinf$ of length $n$. A \pO\/ strategy $\sO$ is a {\bf winning strategy} if $\hOS(\sO, \aHTinf) \in W$ for all $\aHTinf \in \AT^\omega$.

\paragraph*{Definitions on equivalence relations over histories}

Given a game $\langle \AO, \AT, W\rangle$, a {\bf strategy constraint} (constraint for short) is an equivalence relation over histories. Given a constraint $\sim$, a strategy $\sO$ is said to be a $\boldsymbol{\sim}${\bf -strategy} if $\hOS(\sO,\aHT) \sim \hOS(\sO,\aHT') \Rightarrow \sO(\aHT) = \sO(\aHT')$ for all opponent-histories $\aHT,\aHT' \in \AT^*$. Informally, a $\sim$-strategy behaves the same after equivalent histories that are compatible with $\sO$.

Useful predicates on constraints, denoted $\sim$, are defined below.
\begin{enumerate}
\item\label{sim-la} Time awareness: $\aH\sim\aH' \,\Rightarrow\, |\aH| = |\aH'|$, where $|\aH|$ is the length of the sequence/word $\aH$.
\item\label{sim-cas} Closedness by adding a suffix: $\aH \sim \aH'\,\Rightarrow\, \aH\aH'' \sim \aH'\aH''$.
\item Perfect recall: $(\aH \sim \aH' \wedge |\aH| = |\aH'|) \,\Rightarrow\, \forall n \leq |\aH|,\, \aH_{\leq n} \sim \aH'_{\leq n}$

\item\label{sim-wctl} Weak $W$-closedness: $\forall \aHinf,\aHinf' \in (\AO \times \AT)^\omega,(\forall n \in \N, \aHinf_{\leq n} \sim \aHinf'_{\leq n}) \Rightarrow (\aHinf \in W \Leftrightarrow \aHinf' \in W)$

\item\label{cond3} Strong $W$-closedness: $\forall \aHinf,\aHinf' \in (\AO \times \AT)^\omega,$\\
$(\forall n \in \N, \exists \gamma \in (\AO \times \AT)^*, \aHinf_{\leq n}\gamma \sim \aHinf'_{\leq n+|\gamma|}) \Rightarrow (\aHinf \in W \Rightarrow \aHinf' \in W)$
\end{enumerate}
Note that the first three predicates above constrain only (the information available to) the strategies, while the last two predicates constrain also the winning condition.

\paragraph*{Definitions on automata theory}

The {\bf automata} in this article have the classical form $(\Sigma,Q,q_0,\delta)$ where $q_0 \in Q$ and $\delta: Q \times \Sigma \to Q$, possibly with additional accepting states $F \subseteq Q$ in the definition. The state space $Q$ may be infinite, though. The transition function is lifted in two ways by induction. First, to compute the current state after reading a word: $\delta^+(\e) := q_0$ and $\delta^+(ua) := \delta(\delta^+(u), a)$ for all $(u ,a)\in \Sigma^* \times \Sigma$. Second, to compute the sequence of visited states while reading a word: $\delta^{++}(\e) := q_0$ and $\delta^{++}(u a) = \delta^{++}(u)\delta^+(ua)$. Note that $|\delta^{++}(u)| = |u|$ for all $u \in \Sigma^*$.

Given a game $\langle \AO, \AT, W\rangle$, a {\bf memory-aware implementation} of a strategy $\sO$ is a tuple $(M,m_0,\sigma,\mf,)$ where $M$ is a (in)finite set (the memory), $m_0 \in M$ (the initial memory state), $\sigma: M \to \AO$ (the choice of action depending on the memory state), and $\mf : M \times \AT \to M$ (the memory update), such that $\sO = \sigma \circ \mf^+$, where $\mf^+: \AT^* \to M$ (the ``cumulative'' memory update) is defined inductively: $\mf^+(\e) := m_0$ and $\mf^+(\aHT \aT) = \mf(\mf^+(\aHT),\aT)$ for all $(\aHT ,\aT)\in \AT^* \times \AT$. If $M$ is finite, $\sO$ is said to be a {\bf finite-memory strategy}.

{\bf Word pairing}: for all $n \in \N$, for all $u,v \in \Sigma^n$, let ${\bf u \boldsymbol{ \|} v }:= (u_1,v_1) \dots (u_n,v_n) \in (\Sigma^2)^n$.

A time-aware constraint is {\bf $2$-tape-recognizable} using memory states $Q$, if there is an automaton $((\AO \times \AT)^2,Q,q_0,F,\delta)$ such that $u \sim v$ iff $\delta^+(u \| v) \in F$. (It implies $q_0 = \delta^+(\e \| \e) \in F$.) If moreover $Q$ is finite, the constraint is said to be {\bf $2$-tape-regular}. Recognition of relations by several tapes was studied in, e.g., \cite{RS59}. Note that $2$-tape regularity of $\sim$ was called indistinguishability-based in \cite{BD18}.

\paragraph*{Main results}


Let us recall additional notions first. Two functions of domain $\Sigma^*$ that coincide on inputs of length less than $n$ but differ for some input of length $n$ are said to be at distance $\frac{1}{2^n}$. In this context, a map from strategies to strategies (or to $\AT^* \to \AT^*$) is said to be $1$-Lipschitz continuous if from any input strategy that is partially defined for opponent-histories of length up to $n$, one can partially infer the output strategy for opponent-histories of length up to $n$.

\begin{theorem}\label{thm:main}
 Consider a game $\langle \AO, \AT, W\rangle$ and a constraint $\sim$ that is time-aware and closed by adding a suffix. The two results below are independent.
 \begin{enumerate}
 \item\label{thm:main1} {\bf Stronger assumptions and conclusions}: If $\sim$ is also perfectly recalling and weakly $W$-closed, there exists a map $f: (\AT^* \to \AO)  \to (\AT^* \to \AT^*)$ satisfying the following.
 \begin{enumerate}
 \item\label{thm:main10} For all $\sO: \AT^* \to \AO$ let $f_{\sO}$ denote $f(\sO)$; we have $\sO \circ f_{\sO}$ is a $\sim$-strategy, and $\sO \circ f_{\sO}$ is winning if $\sO$ is winning.
 
   \item\label{thm:main11} The map $f$ is $1$-Lipschitz continuous. 
  \item\label{thm:main12} For all $\sO$ the map $f_{\sO}$ preserves the length and the prefix relation.
   \item\label{thm:main13} The map $\sO \mapsto \sO \circ f_{\sO}$ is $1$-Lipschitz continuous. 
  \item\label{thm:main14} If $\AT$ is finite and $\sim$ is computable, $f$ is also computable; and as a consequence, so is $\sO \circ f_{\sO}$ for all computable $\sO$.
  
    \item\label{thm:main15} If $\sO$ has a memory-aware implementation using memory states $M$, and if $\sim$ is $2$-tape recognizable by an automaton with accepting states $F$, then $\sO \circ f_{\sO}$ has a memory-aware implementation using memory states $M \times F$.
 \end{enumerate}

 \item\label{thm:main2} 
 {\bf Weaker assumptions (no perfect recall) and conclusions}: If $\sim$ is also strongly $W$-closed, there exists a self-map of the \pO\/ strategies that satisfies the following.
  \begin{enumerate}
 \item\label{thm:main20} It maps strategies to $\sim$-strategies, and winning strategies to winning strategies.
 
 \item\label{thm:main21} It is $1$-Lipschitz continuous.

 \item\label{thm:main22} If $\AT$ is finite and $\sim$ is computable, the self-map is also computable.
  \end{enumerate}
Moreover, if $\sim$ is $2$-tape recognizable using memory $M_{\sim}$, and if there is a winning strategy with memory $M_\sO$, there is also a winning $\sim$-strategy using memory $\mathcal{P}(M_\sO \times M_\sim)$.

\end{enumerate}
\end{theorem}

 Lemma~\ref{lem:constraint-algebra} below shows that the five constraint predicates behave rather well algebraically. It will be especially useful when handling Boolean combinations of winning conditions.
 
\begin{lemma}\label{lem:constraint-algebra}

Let $\AO$, $\AT$, and $I$ be non-empty sets. For all $i \in I$ let $W_i \subseteq (\AO \times \AT)^\omega$ and $\sim_i$ be a constraint over $(\AO \times \AT)^*$. Let $\sim$ be another constraint.
\begin{enumerate}
 \item\label{lem:constraint-algebra1} If $\sim_i$ is time-aware (resp. closed by adding a suffix, resp. perfectly recalling) for all $i \in I$, so is $\cap_{i \in I}\sim_i$.
 
 \item\label{lem:constraint-algebra2} If $\sim_i$ is weakly (strongly) $W_i$-closed for all $i \in I$, then $\cap_{i \in I}\sim_i$ is weakly (strongly) $\cap_{i \in I}W_i$-closed.
 

\item\label{lem:constraint-algebra4} If $\sim$ is weakly (strongly) $W_i$-closed for all $i \in I$, then $\sim$ is weakly (strongly) $\cup_{i \in I}W_i$-closed.
 

\end{enumerate}
\end{lemma}

\paragraph*{Comments on the definitions and results}

In the literature, \pO\/ {\bf strategies} sometimes have type $(\AO \times \AT)^* \to \AO$. In this article, they have type $\AT^* \to \AO$ instead. Both options would work here, but the latter is simpler.

In the literature, \pO\/ {\bf winning strategies} are sometimes defined as strategies winning against all \pT\/ strategies. In this article, they win against all opponent runs instead. Both options would work here, but the latter is simpler.

Consider a game $\langle \AO, \AT, W\rangle$ and its {\bf sequentialized version} where \pO\/ plays first at each round. It is well-known and easy to show that a \pO\/ strategy wins the {\bf concurrent version} iff she wins the sequential version. I have two reasons to use concurrent games here, though. First, the notation is nicer for the purpose at hand. Second, concurrency does not rule out (semi-)deterministic determinacy of interesting classes of games as in \cite{AAH19} and \cite{SLR18}, and using a sequentialized version of the main result to handle these concurrent games would require cumbersome back-and-forth game sequentialization that would depend on the winner. That being said, many examples in this article are, morally, sequential/turn-based games.

{\bf Strong $W$-closedness} is indeed stronger than {\bf weak $W$-closedness}, as will be proved. Besides these two properties, which relate $\sim$ and $W$, the other predicates on $\sim$ alone are classical when dealing with information imperfectness (possibly known under different names).

However strong the {\bf strong $W$-closedness} may seem, it is strictly weaker than the conjunction of {\bf  perfect recall} and {\bf weak $W$-closedness}, as will be proved. This justifies the attributes stronger/weaker assumptions in Theorem~\ref{thm:main}. Note that the definition of strong $W$-closedness involves only the implication $\aHinf \in W \Rightarrow \aHinf' \in W$, as opposed to an equivalence.

The update functions of memory-aware implementations have type $M \times \AT \to M$, so, informally, they observe only the memory internal state and the opponent's action. In particular they do not observe for free any additional state of some system.

The notion of {\bf $2$-tape recognizability} of equivalence relations is natural indeed, but so is the following. An equivalence relation $\sim$ over $\Sigma^*$ is said to be {\bf $1$-tape recognizable} using memory $Q$ if there exists an automaton $(\Sigma,Q,q_0,\delta)$ such that $u \sim v$ iff $\delta^+(u) = \delta^+(v)$. In this case there are at most $|Q|$ equivalence classes. If $Q$ is finite, $\sim$ is said to be {\bf $1$-tape regular}. When considering time-aware constraints, $2$-tape recognizability is strictly more general, as will be proved, and it yields more general results. A detailed discussion can be found in \cite{BD18}.

Here, the two notions of {\bf recognizability} require nothing about the cardinality of the state space: what matters is the (least) cardinality that suffices. The intention is primarily to invoke the results with finite automata, but allowing for infinite ones is done at no extra cost.

In the {\bf memory} part of Theorem~\ref{thm:main}.\ref{thm:main1}, the Cartesian product involves only the accepting states $F$, but it only spares us one state: indeed, in a automaton that is $2$-tape recognizing a perfectly recalling $\sim$, the non-final states can be safely merged into a trash state. In Theorem~\ref{thm:main}, however, the full $M_{\sim}$ is used and followed by a powerset construction. So by Cantor's theorem, the memory bound in Theorem~\ref{thm:main}.\ref{thm:main2} increases strictly, despite finiteness assumption for $\AT$.

Generally speaking, {\bf $1$-Lipschitz continuous} functions from infinite words to infinite words correspond to (real-time) reactive systems; {\bf continuous} functions correspond to reactive systems with unbounded delay; and {\bf computable functions} to reactive systems with unbounded delay that can be implemented via Turing machines. Therefore both $1$-Lipschitz continuity and computabilty are desirable over continuity (and both imply continuity). 

In Theorem~\ref{thm:main}.\ref{thm:main1}, the derived $\sim$-strategy is of the form $\sO \circ f_{\sO}$, i.e. it is essentially the original $\sO$ fed with modified inputs, which are called virtual opponent-histories. Theorem~\ref{thm:main}.\ref{thm:main11} means that it suffices to know $\sO$ for opponent-history inputs up to some length to infer the corresponding virtual history map $f_{\sO}$ for inputs up to the same length. Theorem~\ref{thm:main}.\ref{thm:main12} means that for each fixed $\sO$, the virtual opponent-history is extended incrementally as the opponent-history grows. The assertations \ref{thm:main11} and \ref{thm:main12} do not imply one another a priori, but that they both hold implies Theorem~\ref{thm:main}.\ref{thm:main13} indeed; and Theorem~\ref{thm:main}.\ref{thm:main13} means that one can start synthesizing a $\sim$-strategy and playing accordingly on inputs up to length $n$ already when knowing $\sO$ on inputs up to length $n$. This process is even computable in the setting of Theorem~\ref{thm:main}.\ref{thm:main14}.

In Theorem~\ref{thm:main}.\ref{thm:main2}, the derived $\sim$-strategy has a very similar form, but the $f_{\sO}$ no longer preserves the prefix relation since the perfect recall assumption is dropped. As a consequence, the virtual opponent-history can no longer be extended incrementally: backtracking is necessary. Thus there is no results that correspond to Theorems~\ref{thm:main}.\ref{thm:main11} and \ref{thm:main}.\ref{thm:main12}, yet one retains both $1$-Lipschitz continuity of the self-map and its computability under suitable assumptions: Theorems~\ref{thm:main}.\ref{thm:main21} and \ref{thm:main}.\ref{thm:main22} correspond to Theorems~\ref{thm:main}.\ref{thm:main13} and \ref{thm:main}.\ref{thm:main14}, respectively.

In Lemma~\ref{lem:constraint-algebra}, {\bf constraints intersection} makes sense since the intersection of equivalence relations is again an equivalence relation. This is false for unions; furthermore, taking the equivalence relation generated by a union of equivalence relations would not preserve weak or strong $W$-closedness.

\section{Application to concurrent games with states and colors}\label{sect:def-cor-sc}

It is sometimes convenient, for intuition and succinctness, to define a winning condition not as a subset of the runs, but in several steps via states and colors. Given the current state, a pair of actions chosen by the players produces a color and determines the next state, and so on. The winning condition is then defined in terms of infinite sequences of colors.

\begin{definition}
An initialized arena is a tuple $\langle \AO, \AT, Q, q_0, \delta, \Col, \tr\rangle$ such that
\begin{itemize}
\item $\AO$ and $\AT$ are non-empty sets (of actions of \pO\/ and \pT\/),
\item $Q$ is a non-empty set (of states),
\item $q_0 \in Q$ (is the initial state),
\item $\delta: Q \times \AO \times \AT \to Q$ (is the state update function).
\item $\Col$ is a non-empty set (of colors),
\item $\tr : Q \times  \AO \times \AT \to \Col$ (is a coloring function).
\end{itemize}
Providing an arena with some $W \subseteq \Col^\omega$ (a winning condition for \pO\/) defines a game.
\end{definition}
In such a game, a triple in $Q \times \AO \times \AT$ is informally called an edge because it leads to a(nother) state via the udpate function $\delta$. Between two states there are $|\AO \times \AT|$ edges. Note that the colors are on the edges rather than on the states. This is generally more succinct and it is strictly more expressive in the following sense: in an arena with finite $Q$, infinite $\AO$ or $\AT$, and colors on the edges, infinite runs may involve infinitely many colors. However, it would never be the case if colors were on the states.

The coloring function $\tr$ is naturally extended to finite and infinite sequences over $\AO \times \AT$. By induction, $\tr^{++}(\e) := \e$, and $\tr^{++}(\aO,\aT) := \tr(q_0,\aO,\aT)$, and $\tr^{++}(\aH(\aO,\aT)) := \tr^{++}(\aH)\tr(\delta^+(\aH),\aO,\aT)$. Then $\tr^{\infty}(\aHinf)$ is the unique sequence in $\Col^\omega$ such that $\tr^{++}(\aHinf_{\leq n})$ is a prefix of $\tr^{\infty}(\aHinf)$ for all $n \in \N$. Note that $|\tr^{++}(\aH)| = |\aH|$ for all $\aH \in (\AO \times \AT)^*$.

The histories, strategies, and winning strategies of the game with states and colors are then defined as these of $\langle \AO, \AT, (\tr^{\infty})^{-1}[W] \rangle$, which is a game as defined in Section~\ref{sect:main-def-res}. Conversely, a game $\langle \AO, \AT, W \rangle$ may be seen as a game with states and colors $ \langle  \AO, \AT, Q, q_0, \delta, \Col, \tr,W\rangle$ where $\Col = \AO \times \AT$, and $Q = \{q_0\}$, and $\tr(q_0,\aO,\aT) = (\aO,\aT)$ for all $(\aO,\aT) \in \AO \times \AT$.

Recall that the update functions of memory-aware implementations have type $M \times \AT \to M$, so they do not observe the states in $Q$ for free. This difference with what is customary in some communities is harmless in terms of finiteness of the strategy memory, though.

\paragraph*{A universal result for concurrent games}

Corollary~\ref{cor:main} below considers games with states and colors. Corollary~\ref{cor:main}.\ref{cor:main1} (resp. \ref{cor:main}.\ref{cor:main2}) is a corollary of Theorem~\ref{thm:main}.\ref{thm:main1} (resp. Theorem~\ref{thm:main}.\ref{thm:main2}). It says that if there is a winning strategy, there is also one that behaves the same after histories of pairs of actions that yield  the same sequence of states (resp. the same current state) and the same sequence of colors. Note that no assumption is made on the winning condition in Corollary~\ref{cor:main}: \emph{it need not be even Borel}.

\begin{corollary}\label{cor:main}
Consider a game with states and colors $G = \langle \AO, \AT, Q, q_0, \delta, \Col, \tr, W\rangle$ such that \pO\/ has a winning strategy $\sO$. The two results below are independent.
\begin{enumerate}
 \item\label{cor:main1} Then she has a winning strategy $\sO'$ (obtained in a Lipschitz manner from $\sO$) that satisfies
\[
\forall \aHT,\aHT' \in \AT^*,\quad \delta^{++} \circ \hOS(\sO',\aHT) = \delta^{++} \circ \hOS(\sO',\aHT) \,\wedge\, \tr^{++} \circ \hOS(\sO',\aHT) = \tr^{++} \circ \hOS(\sO',\aHT')  \quad\Rightarrow\quad \sO'(\aHT) = \sO'(\aHT')
\]
 Furthermore, if $\sO$ can be implemented via memory space $M$, so can $\sO'$; and if $\AT$ is finite and $\sim$ is computable, $\sO'$ is obtained in a computable manner from $\sO$.
 
 \item\label{cor:main2} Then she has a winning strategy $\sO'$ (obtained in a Lipschitz manner from $\sO$) that satisfies
\[
\forall \aHT,\aHT' \in \AT^*,\quad \delta^{+} \circ \hOS(\sO',\aHT) = \delta^{+} \circ \hOS(\sO',\aHT) \,\wedge\, \tr^{++} \circ \hOS(\sO',\aHT) = \tr^{++} \circ \hOS(\sO',\aHT')  \quad\Rightarrow\quad \sO'(\aHT) = \sO'(\aHT')
\]
Furthermore, if $\sO$ can be implemented via memory $M$, then $\sO'$ can be implemented via memory size $2^{|M|(|Q|^2+1)}$; and if $\AT$ is finite and $\sim$ is computable, $\sO'$ is obtained in a computable manner from $\sO$.
\end{enumerate} 
\end{corollary}
On the one hand, Corollary~\ref{cor:main} exemplifies the benefit of dropping  the perfect recall assumption to obtain winning strategies that are significantly more uniform. On the other hand, it exemplifies the memory cost of doing so, which corresponds to the proof-theoretic complexification from Theorem~\ref{thm:main}.\ref{thm:main1} to Theorem~\ref{thm:main}.\ref{thm:main2}, as will be discussed in later sections.

To prove Corollary~\ref{cor:main}.\ref{cor:main2} directly, a natural idea is to ``copy-paste'', i.e., rewrite the strategy at equivalent histories. If done finitely many times, it is easy to prove that the derived strategy is still winning, but things become tricky if done infinitely many times, as it should.

Note that in Corollary~\ref{cor:main}, assumptions and conclusions apply to both players: indeed, since no assumption is made on $W$, its complement satisfies all assumptions, too.

A consequence of Corollary~\ref{cor:main} is that one could define \emph{state-color strategies} as functions in $(Q^* \times C^*) \to \AO$ or even $(Q \times C^* )\to \AO$, while preserving existence of winning strategies. How much one would benefit from doing so depends on the context.

In the remainder of this section, only the weaker sufficient condition, i.e., Theorem~\ref{thm:main}.\ref{thm:main2} is invoked instead of Theorem~\ref{thm:main}.\ref{thm:main1}.

\paragraph*{Between B{\"u}chi and Muller}

In Corollary~\ref{cor:main} the exact sequence of colors mattered, but in some cases from formal methods, the winning condition is invariant under shuffling of the color sequence. Corollary~\ref{cor:shuffle} below provides an example where Theorem~\ref{thm:main} applies (but only to \pO\/). Note that the games defined in Corollary~\ref{cor:shuffle} are a subclass of the concurrent Muller games, where finite-memory strategies suffice \cite{GH82}, and a superclass of the concurrent B{\"u}chi games, where positional (aka memoryless) strategies suffice. In this intermediate class from Corollary~\ref{cor:shuffle}, however, positional strategies are not sufficient: indeed, consider a three-state one-player game where $q_1$ and $q_2$ must be visited infinitely often and where $q_0$ lies between them. Thus, Corollary~\ref{cor:shuffle} is not a corollary of well-known results.

\begin{corollary}\label{cor:shuffle}
Consider a game with states and colors $G = \langle \AO, \AT, Q, q_0, \delta, \Col, \tr, W\rangle$ with finite $Q$ and $C$, and where $W$ is defined as follows: let $(C_i)_{i \in I}$ be subsets of $C$, and let $\boldsymbol{\gamma} \in W$ if there exists $i \in I$ such that all colors in $C_i$ occur infinitely often in $\boldsymbol{\gamma}$.

If \pO\/ has a winning strategy, she has a finite-memory one that behaves the same if the current state and the multiset of seen colors are the same.
\end{corollary}

\paragraph*{Energy games}

The energy winning condition relates to real-valued colors. It requires that at every finite prefix of a run, the sum of the colors seen so far is non-negative. More formally, $\forall \aHinf \in (\AO \times \AT)^\omega,\, \aHinf \in W\, \Leftrightarrow\, \forall n \in \N,\, 0 \leq \sum \tr^{++}(\aHinf_{\leq n})$.

Corollary~\ref{cor:energy} is weaker than the well-known positional determinacy of turn-based energy games, but its proof will be reused in that of Corollary~\ref{cor:muller-energy}.

\begin{corollary}\label{cor:energy}
If \pO\/ has a winning strategy in an energy game $G = \langle \AO, \AT, Q, q_0, \delta, \R, \tr, W\rangle$, she has one that behaves the same if the current time, state, and energy level are the same.
\end{corollary}

\paragraph*{Conjunction of winning conditions}

Corollary~\ref{cor:muller-energy} below strengthens Corollary~\ref{cor:shuffle} (finite-memory aside) by considering the conjunction of the original Muller condition and the energy condition, which works out by Lemma~\ref {lem:constraint-algebra}.\ref{lem:constraint-algebra2}.

\begin{corollary}\label{cor:muller-energy}
 Consider a game with states and colors $G = \langle \AO, \AT, Q, q_0, \delta, \Col \times \R, \tr, W\rangle$ with finite $Q$ and $C$ and where $W \subseteq (\Col \times \R)^\omega$ is defined as follows: let $(C_i)_{i \in I}$ be subsets of $C$, and let $\boldsymbol{\gamma} \in W$ if there exists $i \in I$ such that all colors in $C_i$ occur infinitely often in $\pi_1(\boldsymbol{\gamma})$ (the sequence of the first components) and if the energy level (on the second component) remains non-negative throughout the run, i.e. $\sum \pi_2 \circ \tr^{++}(\boldsymbol{\gamma}_{\leq n} )$ for all $n \in \N$.

If \pO\/ has a winning strategy, she has one that behaves the same if the current state, the multiset of seen colors, and the energy level are the same.
\end{corollary}

The diversity of the above corollaries which follow rather easily from Theorem~\ref{thm:main}, especially from Theorem~\ref{thm:main}.\ref{thm:main2}, should suggest its potentiel range of application.



\section{Variant with stronger assumptions and conclusions}\label{sect:sasc}

This section proves Theorem~\ref{thm:main}.\ref{thm:main1} and Corollary~\ref{cor:main}.\ref{cor:main1}, after a quick reminder about orders.


{\bf Total orders:} A strict total/linear order is a transitive, total, irreflexive binary relation. Such relations are represented by $<$. The reflexive closure of $<$ is denoted $\leq$.

{\bf Well-orders:} A well-order is a strict total order $<$ on some set $\Sigma$ such that every non-empty subset of $\Sigma$ has a mimimum with respect to $<$. Equivalently, it is a strict total order with no infinite descending chain $x_0 > x_1 > x_2 > \dots$.

{\bf Lexicographic orders:} Let $<$ be a strict total order over a set $\Sigma$. The lexicographic extensions of $<$ to $\Sigma^*$ are defined by induction as follows: $<_{\lex}^0$ is the empty relation over $\Sigma$, and for all $n \in \N$, $u,v \in \Sigma^n$, and $x,y \in \Sigma$,
\begin{align}
ux <_{\lex}^{n+1} vy \quad\stackrel{def}{\Longleftrightarrow}\quad u <_{\lex}^n v\quad \vee\quad  (u = v \wedge x < y)
\end{align}
It is well-known that each $<_{\lex}^n$ is a strict total order, and even a well-order if $<$ is a well-order. In this article only words of equal length will be compared lexicographically, so one may write $<_{\lex}$ instead of $<_{\lex}^n$ when the context is clear.


\begin{proof}[Proof of Theorem~\ref{thm:main}.\ref{thm:main1}]
Let $\sO$ be a \pO\/ strategy. Let $<$ be a well-order over $\AT$ (by the well-ordering principle for infinite $\AT$). For all $(\aHT,\aT) \in \AT^* \times \AT$ let $c(\aHT,\aT) := \min_<\{c \in \AT \mid \hOS(\sO,\aHT\aT) \sim \hOS(\sO,\aHT c)\}$. (The set $\{c \in \AT \mid \hOS(\sO,\aHT\aT) \sim \hOS(\sO,\aHT c)\}$ is non-empty, as witnessed by $\aT$, so it has a $<$-minimum since $<$ is a well-order.) Informally, if \pT\/ has played $\aHT \aT$, \pO\/ may pretend that he played $\aHT c(\aHT,\aT)$ instead; this is a harmless approximation up to $\sim$, under the condition that \pO\/ plays according to $\sO$. Note that
\begin{align}
\forall (\aHT,\aT,\aT')  \in \AT^* \times \AT \times \AT,\quad \hOS(\sO,\aHT\aT) \sim \hOS(\sO,\aHT \aT') \Rightarrow c(\aHT,\aT) = c(\aHT,\aT') \label{impl:h-c}
\end{align}
By (\ref{impl:h-c}) above, the function $c$ provides a representative, called a virtual action, for all actions yielding $\sim$-equivalent histories after $\aHT$. It is used incrementally below to provide virtual opponent-histories, i.e. representatives for opponent-histories.

Let $f_{\sO}: \AT^* \to \AT^*$ be defined by induction by $f_{\sO}(\e) := \e$ and $f_{\sO}(\aHT \aT) := f_{\sO}(\aHT) c(f_{\sO}(\aHT),\aT)$. Informally, if \pT\/ has played $\aHT$, \pO\/ will pretend that he played $f_{\sO}(\aHT)$ instead, and she will play accordingly, using the same strategy $\sO$. Said otherwise, the new \pO\/ strategy is $\sO \circ f_{\sO}$. 

{\bf $1$-Lipschitz continuity and computability: } The definition of $f_{\sO}$ for inputs of length up to $n$ involves only outputs of $\sO$ for inputs of length up to $n$, via $\hOS$ used in $c$. So the map $\sO \mapsto f_{\sO}$ is $1$-Lipschitz continuous, thus proving (\ref{thm:main11}). Moreover, from its inductive definition, $f_{\sO}$ clearly preserves the length and the prefix relation, thus proving (\ref{thm:main12}). As a consequence of the two facts above, the map $\sO \mapsto \sO \circ f_{\sO}$ is also $1$-Lipschitz continuous, i.e. if one knows $\sO$ for inputs of length up to $n$, one also knows $\sO \circ f_{\sO}$ for inputs of length up to $n$, thus proving (\ref{thm:main13}). If $\AT$ is finite, the order $<$ is computable. If in addition $\sim$ is computable, so is the function $c$, and subsequently the maps $\sO \mapsto f_{\sO}$ and $\sO \mapsto \sO \circ f_{\sO}$, thus proving (\ref{thm:main14}).

{\bf Intermediate results:} The formulas (\ref{equi-f-b}) and (\ref{circ-f-f}) below will be used several times in this proof.
\begin{align}
\forall(\aHT,\aT) \in \AT^* \times \AT, \quad \hOS(\sO,f_{\sO}(\aHT)\aT) & \sim \hOS(\sO,f_{\sO}(\aHT)c(f_{\sO}(\aHT),\aT)) \mbox{ by definition of } c,\nonumber\\
  & = \hOS(\sO,f_{\sO}(\aHT \aT)) \mbox{ by ``folding'' the definition of }f_{\sO}.\label{equi-f-b}
\end{align}
Let us prove (\ref{circ-f-f}) below by induction on $\aHT$.
\begin{align}
\forall \aHT \in \AT^*,\quad \hOS(\sO \circ f_{\sO},\aHT) \sim \hOS(\sO, f_{\sO}(\aHT)) \label{circ-f-f}
\end{align}
For the base case, trivially $\e = \e$. For the inductive case, 
\begin{align*}
\hOS(\sO \circ f_{\sO},\aHT \aT) & = \hOS(\sO \circ f_{\sO},\aHT) (\sO \circ f_{\sO} (\aHT),\aT) \mbox{ by unfolding the definition of }\hOS,\\
    & \sim \hOS(\sO, f_{\sO}(\aHT)) (\sO \circ f_{\sO} (\aHT),\aT) \mbox{ by I.H. and closedness by adding a suffix},\\
    & = \hOS(\sO, f_{\sO}(\aHT)) (\sO(f_{\sO}(\aHT)),\aT) \mbox{ by definition of function composition},\\
    & = \hOS(\sO, f_{\sO}(\aHT) \aT) \mbox{ by folding the definition of }\hOS.
\end{align*}
One completes the induction step by combining the above $\sim$-equivalence with Formula (\ref{equi-f-b}).

{\bf Proof that $\sO \circ f_{\sO}$ is winning if $\sO$ is:} Since $f_{\sO}$ preserves the length and the prefix relation, one can extend its domain and codomain to $\AT^\omega$: for all $\aHTinf \in \AT^\omega$ let $f_{\sO}(\aHTinf)$ be the only sequence $\aHTinf' \in \AT^\omega$ such that $\aHTinf'_{\leq n} = f_{\sO}(\aHTinf_{\leq n})$ for all $n \in \N$. For all $\aHTinf \in \AT^\omega$ and $n \in \N$
\begin{align*}
\hOS(\sO \circ f_{\sO}, \aHTinf)_{\leq n} & =  \hOS(\sO \circ f_{\sO}, \aHTinf_{\leq n})\mbox{ by preservation of length and of the prefix relation}\\
  & \sim  \hOS(\sO, f_{\sO}(\aHTinf_{\leq n})) \mbox{ by Formula (\ref{circ-f-f})},\\
  & =  \hOS(\sO, f_{\sO}(\aHTinf)_{\leq n}) = \hOS(\sO, f_{\sO}(\aHTinf))_{\leq n}.
\end{align*}
If $\sO$ is winning, $\hOS(\sO, f_{\sO}(\aHTinf)) \in W$, and $\hOS(\sO \circ f_{\sO}, \aHTinf) \in W$ by weak $W$-closedness, so $\sO \circ f_{\sO}$ is also winning.

{\bf Proof that $\sO \circ f_{\sO}$ is a $\sim$-strategy:} it suffices to prove the following, by induction on $\aHT$.
\[
\forall \aHT,\aHT' \in \AT^*, \quad \hOS(\sO \circ f_{\sO},\aHT) \sim \hOS(\sO \circ f_{\sO},\aHT') \Rightarrow f_{\sO}(\aHT) = f_{\sO}(\aHT')
\]
Then from the conclusion $f_{\sO}(\aHT) = f_{\sO}(\aHT')$ one obtains $\sO \circ f_{\sO}(\aHT) = \sO \circ f_{\sO}(\aHT')$. For the base case, trivially $f_{\sO}(\e) = f_{\sO}(\e)$. For the inductive case, let us assume that $\hOS(\sO \circ f_{\sO},\aHT \aT) \sim \hOS(\sO \circ f_{\sO},\aHT' \aT')$, so $\hOS(\sO \circ f_{\sO},\aHT)(\sO \circ f_{\sO}(\aHT),\aT) \sim \hOS(\sO \circ f_{\sO},\aHT')(\sO \circ f_{\sO}(\aHT'), \aT')$ by unfolding the definition of $\hOS$, so $\hOS(\sO \circ f_{\sO},\aHT) \sim \hOS(\sO \circ f_{\sO},\aHT')$ by perfect recall, so $f_{\sO}(\aHT) = f_{\sO}(\aHT')$ by induction hypothesis. The assumption $\hOS(\sO \circ f_{\sO},\aHT \aT) \sim \hOS(\sO \circ f_{\sO},\aHT' \aT')$ implies $\hOS(\sO,f_{\sO}(\aHT) \aT) \sim \hOS(\sO, f_{\sO}(\aHT) \aT')$ by (\ref{circ-f-f}) and (\ref{equi-f-b}), and by rewriting $f_{\sO}(\aHT')$ with $f_{\sO}(\aHT)$. So $c(f_{\sO}(\aHT),\aT) = c(f_{\sO}(\aHT),\aT')$ by (\ref{impl:h-c}). Therefore $f_{\sO}(\aHT \aT) = f_{\sO}(\aHT)c(f_{\sO}(\aHT),\aT) = f_{\sO}(\aHT')c(f_{\sO}(\aHT),\aT') = f_{\sO}(\aHT' \aT')$, thus completing the induction.

{\bf The memory-aware implementation of $\sO \circ f_{\sO}$:} Let us assume that some $(M,m_0,\sigma,\mf)$ is a memory-aware implementation of $\sO$, and that $\sim$ is recognized by an automaton $\mathcal{A} = ((\AO \times \AT)^2,Q,q_0,F,\delta)$. A natural idea is, as $\aHT$ is being incrementally extended by \pT, to keep track of both (interdependently) the strategy memory state via $\mf$, and $f_{\sO}(\aHT)$ via $\mathcal{A}$, and to feed the former to $\sigma$.  

Let $\overline{m}_0 := (m_0,q_0) \in M \times F$, which will be the new initial memory state. Let $\overline{M} := M \times F$, and for all $(m,q) \in M \times F$ let $\overline{\sigma}(m,q) = \sigma(m)$. For the memory update, let 
\[
\begin{array}{rrcl}
\overline{\mf} : & \overline{M} \times \AT & \longrightarrow & \overline{M}\\
    & (m,q,\aT) & \longmapsto & (\mf(m,c'), \delta(q,(\sigma(m),c'),(\sigma(m),c'))\\
    &&& \mbox{ where }c' = c'(m,q,\aT) = \min_{<}\{c \in \AT \mid \delta(q,(\sigma(m),b),(\sigma(m),c)) \in F\}
\end{array}
\]
To see that $\overline{\mf}$ is well-defined, check that $c'(m,q,\aHT)$ is well-defined: $\delta(q,(\sigma(m),b),(\sigma(m),b)) \in F$ due to closedness by adding a suffix and since $q \in F$ by definition of $\overline{M}$, so the set at hand is non-empty and has a $<$-minimum since $<$ is a well-order. 

The aim is to prove that $\sO \circ f_{\sO} = \overline{\sigma} \circ \overline{\mf}^+$. The main step is to prove the equation below, by induction on $\aHT$. Then, it suffices to compose both sides of the equation with $\overline{\sigma}$, to invoke $\overline{\sigma}(m,q) = \sigma(m)$ on the right-hand side, and to recall that $\sO = \sigma \circ \mf^+$ by definition.
\[
\forall \aHT \in \AT^*, \quad \overline{\mf}^+(\aHT) = \big(\mf^+ \circ f_{\sO}(\aHT), \delta^+\big(\hOS(\sO,f_{\sO}(\aHT)) \| \hOS(\sO,f_{\sO}(\aHT))\big)\big)
\]
Base case: $\overline{\mf}^+(\e)  = \overline{m}_0 = (m_0,q_0) = (\mf^+(\e), \delta^+(\e)) = \big(\mf^+ \circ f_{\sO}(\e), \delta^+\big(\hOS(\sO,f_{\sO}(\e)) \| \hOS(\sO,f_{\sO}(\e))\big)\big)$. For the inductive case, let us first note that for all $(\aHT,\aT,c) \in \AT^* \times \AT \times \AT$,
\begin{align}
&\delta^+\big(\hOS(\sO,f_{\sO}(\aHT)\aT) \| \hOS(\sO,f_{\sO}(\aHT)c)\big) \nonumber\\
& = \quad \delta\big(\delta^+\big(\hOS(\sO,f_{\sO}(\aHT)) \| \hOS(\sO,f_{\sO}(\aHT))\big), (\sO \circ f_{\sO}(\aHT),\aT), (\sO \circ f_{\sO}(\aHT),c)\big) \nonumber\\
& \quad \quad \mbox{ by unfolding the definition of }\hOS \mbox{ and by definition of } \delta^+ ,\nonumber\\
& = \quad \Delta(\aT,c) \mbox{ by rewriting twice } \sO \mbox{ with } \sigma \circ \mf^+\mbox{, which is correct by definition},\label{equi.h-delta}\\
& \quad \quad \mbox{ where } \Delta(\aT,c) := \delta\big(\delta^+\big(\hOS(\sO,f_{\sO}(\aHT)) \| \hOS(\sO,f_{\sO}(\aHT))\big), (\sigma \circ \mf^+ \circ f_{\sO}(\aHT),\aT), (\sigma \circ \mf^+ \circ f_{\sO}(\aHT),c)\big).\nonumber
\end{align}
Futhermore, $\hOS(\sO,f_{\sO}(\aHT)\aT) \sim  \hOS(\sO,f_{\sO}(\aHT)c)$ iff $\delta^+\big(\hOS(\sO,f_{\sO}(\aHT)\aT) \|  \hOS(\sO,f_{\sO}(\aHT)c)\big) \in F$ (by definition of the automaton recognizing $\sim$) iff $\Delta(\aT,c)  \in F$ (by Equality (\ref{equi.h-delta})), so $c(f_{\sO}(\aHT),\aT) =$\\$c'(\mf^+ \circ f_{\sO}(\aHT), \delta^+\big(\hOS(\sO,f_{\sO}(\aHT)) \| \hOS(\sO,f_{\sO}(\aHT))\big), \aT)$, which is written $c''$ for short. Therefore,
\begin{align*}
\overline{\mf}^+(\aHT \aT) & = \overline{\mf}(\overline{\mf}^+(\aHT), \aT) \mbox{ by unfolding the definition of }\overline{\mf}^+,\\
  & = \overline{\mf}\big(\mf^+ \circ f_{\sO}(\aHT), \delta^+\big(\hOS(\sO,f_{\sO}(\aHT)) \| \hOS(\sO,f_{\sO}(\aHT))\big), \aT\big) \mbox{ by I.H.},\\
  & = \big(\mf(\mf^+ \circ f_{\sO}(\aHT),c''), \Delta(c'',c'') \big)\mbox{ by definition of }\overline{\mf},\\
  & = \big(\mf^+ (f_{\sO}(\aHT) c'')), \Delta(c'',c'') \big) \mbox{ by folding the definition of }\mf^+,\\
  & = \big(\mf^+ (f_{\sO}(\aHT) c''), \delta^+\big(\hOS(\sO,f_{\sO}(\aHT)c'') \| \hOS(\sO,f_{\sO}(\aHT)c'')\big)), \mbox{ by Equality (\ref{equi.h-delta})},\\
  & = \big(\mf^+ \circ f_{\sO}(\aHT \aT), \delta^+\big(\hOS(\sO,f_{\sO}(\aHT \aT)) \| \hOS(\sO,f_{\sO}(\aHT \aT))\big)) \mbox{ since }c'' = c(f_{\sO}(\aHT),\aT) \mbox{, as argued} \\
  & \quad \quad \mbox {above, and by definition of }f_{\sO}. \mbox{ This completes the induction}.
\end{align*}
\end{proof}

Let us comment on the proof of Theorem~\ref{thm:main}.\ref{thm:main1}. First, the well-order $<$ is only used to define a choice function for all non-empty subsets of $\AT$, i.e. further order-theoretic properties are not used. Second, the key object derived from a strategy $\sO$ is the representative function $f_{\sO}$ that helps build the sought strategy $\sO \circ f_{\sO}$. The property that $\sO \circ f_{\sO}$ is winning (if $\sO$ is) does not rely on $\sim$'s perfectly recalling, but the property that it is a $\sim$-strategy does.

\begin{proof}[Proof of Corollary~\ref{cor:main}.\ref{cor:main1}]
Let $\sO$ be a \pO\/ winning strategy. Let $\sim$ be the equivalence relation over $(\AO \times \AT)^*$ defined by $\aH \sim \aH'$ if $\tr^{++}(\aH) = \tr^{++}(\aH')$ and $\delta^{++}(\aH) = \delta^{++}(\aH')$. This $\sim$ is time-aware since, e.g. $\tr^{++}(\aH) = \tr^{++}(\aH')$ implies $|\aH| = |\aH'|$. This $\sim$ is closed by adding a suffix: let us assume $\aH \sim \aH'$, and let $\aH'' \in (\AO \times \AT)^*$; it is straightforward to show by induction on $\aH''$ that $\aH \aH''\sim \aH' \aH''$. This $\sim$ is perfectly recalling: given $\aH (\aO,\aT)\sim \aH'(\aO',\aT')$, it is straightforward to show that $\aH \sim \aH'$. This $\sim$ is weakly $W$-closed: let $(\aHinf,\aHinf') \in ((\AO \times \AT)^\omega)^2$ and let us assume that $\aHinf_{\leq n} \sim \aHinf'_{\leq n}$ for all $n \in \N$. In particular $\tr^{++}(\aHinf_{\leq n}) = \tr^{++}(\aHinf'_{\leq n})$ for all $n \in \N$, so $\tr^{\infty}(\aHinf) = \tr^{\infty}(\aHinf')$, which implies $\aHinf \in (\tr^{\infty})^{-1}[W] \Leftrightarrow \aHinf' \in (\tr^{\infty})^{-1}[W]$. 

Recall that, by definition, the strategies in $G' = \langle \AO, \AT, (\tr^{\infty})^{-1}[W]\rangle$ are exactly the strategies in $G$, and that a strategy is winning in $G$ iff it is winning in $G'$. So, by Theorem~\ref{thm:main}.\ref{thm:main1} applied to the game $G'$, $\sO \circ f_{\sO}$ is a winning $\sim$-strategy also in $G$.

Furthermore, there is a two-state automaton that recognizes $\sim$, by staying in $\{q_0\} = F$ until states or colors differ . So again by Theorem~\ref{thm:main}.\ref{thm:main1}, $\sO \circ f_{\sO}$ uses the same memory as $\sO$.
\end{proof}


\section{Weaker assumptions and conclusions}\label{sect:wawc}

This section proves the first part of Theorem~\ref{thm:main}.\ref{thm:main2}, i.e. without memory awareness. From now on, perfect recall is no longer assumed. To compensate this, the assumption on weak $W$-closedness will be replaced with strong $W$-closedness. First, Proposition~\ref{prop:finite-wawc} below proves a result for finite $\AT$. This assumption is dropped in the first part of Theorem~\ref{thm:main}.\ref{thm:main2}, whose proof consists in replacing and expanding the last part of the proof of Proposition~\ref{prop:finite-wawc}.

\begin{proposition}\label{prop:finite-wawc}
Consider a game $\langle \AO, \AT, W\rangle$ with finite $\AT$ and a constraint $\sim$ that is time-aware, closed by adding a suffix, and strongly $W$-closed.
Then there exists a self-map on the \pO\/ strategies that is $1$-Lipschitz continuous, maps strategies to $\sim$-strategies, and maps winning strategies to winning strategies.

If $\sim$ is computable, so is the aforementioned self-map.

\begin{proof}
Let $\sO$ be a \pO\/ strategy. Let $<$ be a strict total order over $\AT$, and for all $n \in \N$ let $<_{\lex}^n$ be its lexicographic extensions to $\AT^n$, which is a well-order, as mentioned in Section~\ref{sect:sasc}. Note that $\{\aHT \in \AT^* \mid \hOS(\sO,\aHT) \sim \aH\} \neq \emptyset$ for all $\aH \in \tilde{\hOS}(\sO,\AT^*) := \{\aH \in (\AO \times \AT)^* \mid \exists \aH' \in \hOS(\sO,\AT^*), \aH \sim \aH'\}$, and let $f_{\sO}: \tilde{\hOS}(\sO,\AT^*) \to \AT^*$ be defined by $f_{\sO}(\aH) = \min_{<_\lex^{|\aH|}}\{\aHT \in \AT^* \mid \hOS(\sO,\aHT) \sim \aH\}$. Note that $f_{\sO}$ was define differently in the proof of Theorem~\ref{thm:main}.\ref{thm:main1}, let alone that it had a different domain. The following property is easy to check.
\begin{align}
\forall \aH,\aH' \in \tilde{\hOS}(\sO,\AT^*), \quad\hOS(\sO,f_{\sO}(\aH))  \sim \aH  \quad\wedge \quad(\aH \sim \aH' \Rightarrow f_{\sO}(\aH) = f_{\sO}(\aH')) \label{prop:sim-finite-action}
\end{align}

Let us define a strategy $\sO'$ such that $\hOS(\sO',\AT^*) \subseteq \tilde{\hOS}(\sO,\AT^*)$ by induction on its argument: base case, note that $\hOS(\sO',\e) = \e = \hOS(\sO,\e) \in \tilde{\hOS}(\sO,\AT^*)$ and let $\sO'(\e) := \sO \circ f_{\sO}  \circ \hOS(\sO',\e) = \sO(\e)$. Inductive case, let $(\aHT,\aT) \in \AT^* \times \AT$ be such that $\hOS(\sO',\aHT) \in \tilde{\hOS}(\sO,\AT^*)$ and $\sO'(\aHT) = \sO \circ f_{\sO}  \circ \hOS(\sO',\aHT)$. Then
\begin{align}
\hOS(\sO',\aHT \aT) & = \hOS(\sO',\aHT)(\sO'(\aHT),\aT) \mbox{ by unfolding the definition of }\hOS,\nonumber\\
  & \sim \hOS(\sO,  f_{\sO}  \circ \hOS(\sO',\aHT))(\sO'(\aHT),\aT) \mbox{ by Property (\ref{prop:sim-finite-action}) and closedness by adding a suffix},\nonumber\\
  & = \hOS(\sO, f_{\sO}  \circ \hOS(\sO',\aHT))(\sO \circ f_{\sO}  \circ \hOS(\sO',\aHT),\aT)\mbox{ by rewriting } \sO'(\aHT),\nonumber\\
  & = \hOS(\sO, (f_{\sO}  \circ \hOS(\sO',\aHT)) \aT)\mbox{ by folding the definition of }\hOS.\label{eq:s-s-f}
\end{align}
So $\hOS(\sO',\aHT \aT) \in \tilde{\hOS}(\sO,\AT^*)$ and one can set $\sO'(\aHT \aT) := \sO \circ f_{\sO}  \circ \hOS(\sO',\aHT \aT)$, which is well-defined since $\hOS(\sO',\aHT \aT)$ does not depend on $\sO'(\aHT \aT)$. This completes the induction. Note that the map $\sO \mapsto \sO'$ is $1$-Lipschitz continuous due to its inductive definition. 

Furthermore, $\sO'$ is a $\sim$-strategy: let $\aHT,\aHT' \in \AT^*$ be such that $\hOS(\sO',\aHT) \sim \hOS(\sO',\aHT')$. Then
\begin{align*}
\sO'(\aHT) & = \sO \circ f_{\sO}  \circ \hOS(\sO',\aHT) \mbox{ by definition of }\sO'\\
        & = \sO \circ f_{\sO}  \circ \hOS(\sO',\aHT') \mbox{ by Property (\ref{prop:sim-finite-action})}\\
        & = \sO'(\aHT') \mbox{ again by definition of }\sO'.
\end{align*}

Since $\AT$ is finite, $<$ is computable, and so are its lexicographic extensions. Therefore, if $\sim$ is computable, so is the map $\sO \mapsto \sO'$ via $\hOS$ and the definition of $f_{\sO}$.

Now, let us assume that $\sO$ is winning, i.e. $\hOS(\sO,\AT^\omega) \subseteq W$, consider an arbitrary $\aHTinf' \in \AT^\omega$, and let $\aHinf' := \hOS(\sO',\aHTinf')$. For all $n \in \N$, let $\delta(n) := \hOS(\sO,f_{\sO} \circ \hOS(\sO',\aHTinf'_{\leq n}))$, and consider the tree induced by all the $f_{\sO} \circ \hOS(\sO',\aHTinf'_{\leq n})$ (and their prefixes): it is finitely branching since $\AT$ is finite, so by Koenig's Lemma this tree has an infinite branch. Said otherwise, there is $\aHTinf \in \AT^\omega$ such that for all $n \in \N$ there is $k \in \N$ such that $\aHTinf_{\leq n} = f_{\sO} \circ \hOS(\sO',\aHTinf'_{\leq n+k})_{\leq n}$. Let $\aHinf = \hOS(\sO,\aHTinf)$. So for all $n \in \N$ there is $k \in \N$ such that $\aHinf_{\leq n} = \hOS(\sO,\aHTinf_{\leq n}) = \hOS(\sO, f_{\sO} \circ \hOS(\sO',\aHTinf'_{\leq n+k})_{\leq n}) = \delta(n+k)_{\leq n}$, i.e. there exists $\gamma$ such that $\aHinf_{\leq n}\gamma = \delta(n+|\gamma|)$. Furthermore, $\aHinf'_{\leq n+|\gamma|} = \hOS(\sO',\aHTinf'_{\leq n+|\gamma|}) \sim \delta(n+|\gamma|)$ for all $n \in \N$ by Property (\ref{prop:sim-finite-action}), so $\aHinf_{\leq n}\gamma \sim \aHinf'_{\leq n+|\gamma|}$. Since $\aHinf \in \hOS(\sO,\AT^\omega)  \subseteq W$, also $\aHinf' \in W$ by strong $W$-closedness. Therefore $\sO'$ is winning, since $\aHTinf'$ was arbitrary.

\end{proof}
\end{proposition}

Let us comment on the proof of Proposition~\ref{prop:finite-wawc}. First, like for Theorem~\ref{thm:main}.\ref{thm:main1}, the well-order $<$ is only used to define a choice function, but this time for some non-empty subsets of $\AT^*$ via its well-ordered extensions, namely the $<_{\lex}^n$. Second, $f_{\sO}$ is now defined directly over $\AT^*$, i.e. not in an incremental fashion. Third, the sought strategy $\sO'$ is defined inductively by $\sO'(\aHT \aT) := \sO \circ f_{\sO}  \circ \hOS(\sO',\aHT \aT)$, i.e. not only on \pT\/ action sequence (as in the proof of Theorem~\ref{thm:main}.\ref{thm:main1}), but also on $\sO'$ itself. Fourth, whereas $\hOS(\sO \circ f_{\sO},\AT^*) \subseteq \hOS(\sO, \AT^*)$ in the proof of Theorem~\ref{thm:main}.\ref{thm:main1}, now only $\hOS(\sO',\AT^*) \subseteq  \tilde{\hOS}(\sO, \AT^*)$ holds. Fifth, since the representatives are no longer prefixes of each other, they do no induce a canonical infinite representative. One retrieves one thanks to (weak) Koenig's Lemma, thus incurring a loss in terms of computability and constructivity. Sixth, the strong $W$-closedness seems to be needed, as opposed to the weak one. Indeed, in the induced infinite tree, there is no reason why there should exist an infinite branch with all (or even any) prefixes being original representatives. (E.g. the only infinite branch in the tree induced by the $0^n1$ is $0^\omega$.)

Theorem~\ref{thm:main}.\ref{thm:main2}, whose first part is proved below, drops the finiteness assumption for $\AT$, so the proof can no longer rely on Koenig's Lemma.

\begin{proof}[Proof of Theorem~\ref{thm:main}.\ref{thm:main2} without memory awareness]
The beginning of this proof is like the beginning of the proof of Proposition~\ref{prop:finite-wawc}, apart from the sentence ``Let $<$ be a strict total order over $\AT$'', which is replaced with the sentence ``By the well-ordering principle (equivalent to the axiom of choice), let $<$ be a well-order over $\AT$.'' More importantly, the last paragraph (starting with ``Now, let us assume that $\sO$'') is replaced with the remainder of this proof.

Let $(\aHT,\aT) \in \AT^* \times \AT$. By Property (\ref{eq:s-s-f}), $\aH := \hOS(\sO',\aHT \aT)  \sim \hOS(\sO, (f_{\sO}  \circ \hOS(\sO',\aHT)) \aT)$. So $f_{\sO}(\aH) \leq_{\lex}^{|\aH|} (f_{\sO}  \circ \hOS(\sO',\aHT)) \aT$ by definition of $f_{\sO}$, so $(f_{\sO}  \circ \hOS(\sO',\aHT \aT))_{\leq |\aHT|} \leq_{\lex}^{|\aHT|} f_{\sO}  \circ \hOS(\sO',\aHT)$ by definition of the $<_{\lex}^n$. Therefore $(f_{\sO}  \circ \hOS(\sO',\aHT \aT))_{\leq n} \leq_{\lex}^{n} (f_{\sO}  \circ \hOS(\sO',\aHT))_{\leq n}$ for all $n \leq |\aHT|$.


Next, let $\aHTinf' \in \AT^\omega$. By the above property, for all $n \in \N$ the sequence $(f_{\sO}  \circ \hOS(\sO' ,\aHTinf'_{\leq n+k})_{\leq n})_{k \in \N}$ is $\leq_{\lex}^n$-non-increasing, so it is eventually constant from some index on, since $<_{\lex}^n$ is a well-order. Let $k_n$ be the minimum such index. Thus, $k_n \leq k_{n+1} + 1$ for all $n \in \N$. Let $\aHT(n) := f_{\sO}  \circ \hOS(\sO' ,\aHTinf'_{\leq n+k_n})_{\leq n}$. By $k_n \leq k_{n+1}+1$, we find that $\aHT(n)$ is a prefix of $\aHT(n+1)$ for all $n \in \N$, so the $\aHT(n)$ are the prefixes of a unique $\aHTinf \in \AT^\omega$.

Let $\aHinf := \hOS(\sO,\aHTinf)$, let $\aHinf' := \hOS(\sO',\aHTinf')$, let $\delta(n) := \hOS(\sO, f_{\sO}  \circ \hOS(\sO',\aHTinf'_{\leq n+k_n}))$, and let us show that $\hOS(\sO' ,\aHTinf') \in W$ by invoking strong $W$-closedness. 
On the one hand
\begin{align*}
\delta(n)_{\leq n} & =  \hOS(\sO, f_{\sO}  \circ \hOS(\sO',\aHTinf'_{\leq n+k_n}))_{\leq n} = \hOS(\sO, f_{\sO}  \circ \hOS(\sO',\aHTinf'_{\leq n+k_n})_{\leq n}) \\
    & = \hOS(\sO, \aHT(n)) = \hOS(\sO, \aHTinf_{\leq n})= \hOS(\sO, \aHTinf)_{\leq n} = \aHinf_{\leq n}
\end{align*}
so there exists $\gamma(n)$ such that $\aHinf_{\leq n}\gamma(n) = \delta(n)$, which implies that $|\gamma(n)| = k_n$. On the other hand $\delta(n)  = \hOS(\sO, f_{\sO}  \circ \hOS(\sO',\aHTinf'_{\leq n+k_n})) \sim \hOS(\sO',\aHTinf'_{\leq n+k_n}) = \hOS(\sO',\aHTinf')_{\leq n+k_n} = \aHinf'_{\leq n+|\gamma(n)|}$, where the equivalence holds by Property (\ref{eq:s-s-f}). So $\aHinf' = \hOS(\sO' ,\aHTinf') \in W$ since $\aHinf = \hOS(\sO ,\aHTinf)$ and by strong $W$-closedness. Since $\aHTinf'$ was arbitrary, $\sO'$ is a winning strategy.
\end{proof}

Let us compare the above proof of Theorem~\ref{thm:main}.\ref{thm:main2} without memory awareness with the proof of Proposition~\ref{prop:finite-wawc}. First, it now invokes the axiom of choice, to obtain a well-order. Second, the well-order, which was obtained for free for finite $\AT$, is no longer used only as a choice function: its lexicographic extensions, due to their interaction, ensure a convenient convergence. Third, Property~(\ref{eq:s-s-f}) was not invoked in the last paragraph of the proof of Proposition~\ref{prop:finite-wawc}, but is now involved in its replacement above.

\begin{proof}[Proof of Corollary~\ref{cor:energy}]
Let $\aH \sim \aH'$ iff $|\aH| = |\aH'|$ and $\delta^+(\aH) = \delta^+(\aH')$ and $\sum \tr^{++}(\aH) = \sum \tr^{++}(\aH')$ and for all $n \leq |\aH|$ we have $0 \leq \sum \tr^{++}(\aH_{\leq n})$ iff $0 \leq \sum \tr^{++}(\aH'_{\leq n})$. Let us prove useful properties of $\sim$ below.
\begin{itemize}
 \item Time awareness: by definition.
 
 \item Closedness by adding a suffix: Let $\aH,\aH',\aH'' \in (\AO \times \AT)^*$ be such that $\aH \sim \aH'$. First, $|\aH| = |\aH'|$ implies $|\aH \aH''| = |\aH' \aH''|$; second, $\delta^+(\aH) = \delta^+(\aH')$ implies $\delta^+(\aH \aH'') = \delta^+(\aH' \aH'')$; third, $\sum \tr^{++}(\aH) = \sum \tr^{++}(\aH')$ implies $\sum \tr^{++}(\aH \aH'') = \sum \tr^{++}(\aH'\aH'')$; finally for all $n \leq |\aH \aH''|$ we have $0 \leq \sum \tr^{++}((\aH \aH'')_{\leq n})$ iff $0 \leq \sum \tr^{++}((\aH' \aH'')_{\leq n})$.

 \item Strong $W$-closedness: let $\aHinf,\aHinf' \in (\AO \times \AT)^\omega,$ be such that $(\forall n \in \N, \exists \gamma \in (\AO \times \AT)^*, \aHinf_{\leq n}\gamma \sim \aHinf'_{\leq n+|\gamma|})$. So for all $n \in \N$, we have $0 \leq \sum \tr^{++}(\aHinf_{\leq n})$ iff $0 \leq \sum \tr^{++}(\aHinf'_{\leq n})$, so $\aHinf \in W$ iff $\aHinf' \in W)$.
\end{itemize}
Since \pO\/ has a winning strategy, Theorem~\ref{thm:main}.\ref{thm:main2} implies that she has a winning $\sim$-strategy.

A priori, this $\sim$-strategy depends not only on the current time, state, and energy level, but also on the sign of the energy level on prefixes of the history. However, since it is winning, the energy level is always non-negative, so this $\sim$-strategy is a sought witness.
\end{proof}

\begin{proof}[Proof of Corollary~\ref{cor:muller-energy}]
Let $\sim_1$ and $\sim_2$ be the constraints from Corollaries~\ref{cor:shuffle} and \ref{cor:energy}, respectively, lifted to $(\Col \times \R)^*$. Let $W_1$ and $W_2$ be the winning conditions from Corollaries~\ref{cor:shuffle} and \ref{cor:energy}, respectively, lifted to $(\Col \times \R)^\omega$. The constraints $\sim_1$ and $\sim_2$ are time-aware, closed by adding a suffix, and strongly $W_1$ and $W_2$-closed, respectively, as shown in the proofs of Corollaries~\ref{cor:shuffle} and \ref{cor:energy}. So by Lemma~\ref {lem:constraint-algebra}, the constraint $\sim_1 \cap \sim_2$ is time-aware, closed by adding a suffix, and strongly $(W_1 \cap W_2)$-closed. Therefore, if \pO\/ has a winning strategy, she has a winning $(\sim_1 \cap \sim_2)$-strateggy by Theorem~\ref{thm:main}.\ref{thm:main2}
\end{proof}

\section{Weaker assumptions and conclusions with memory-aware strategies}
\label{sect:memory}

This section proves the part of Theorem~\ref{thm:main}.\ref{thm:main2} with memory awareness, gives an example where uniformization necessarily increases the used memory, and proves Corollaries~\ref{cor:main}.\ref{cor:main2} and \ref{cor:shuffle}.

\begin{proof}[Proof of Theorem~\ref{thm:main}.\ref{thm:main2} with memory awareness]
Let $\sO$ be a \pO\/ strategy with memory-aware implementation $(M,m_0,\sigma,\mf)$. 

{\bf Defining a new strategy:} Let $<$ be a strict total order over $M$. Let $((\AO \times \AT)^2,Q,q_0,F,\delta)$ be an automaton recognizing time-aware $\sim$, and let us define a memory-aware implementation $(\overline{M},\overline{m}_0,\overline{\sigma},\overline{\mf})$ that depends on the above objects. Below $\pi_1$ and $\pi_2$ are the projections on, respectively, the first and second component of a pair. Let $\overline{M} := \{ C \in \mathcal{P}(M \times Q) \mid \pi_2[C] \cap F \neq \emptyset\}$, where 
$\pi_2[C] \cap F \neq \emptyset \Leftrightarrow\exists (x,q) \in C, q \in F$.  Let $\overline{m}_0 := \{(m_0,q_0)\}$. Note that $\overline{m}_0 \in \overline{M}$ since $\pi_2[\{(m_0,q_0)\}] = \{q_0\} \subseteq F$.
\[
\begin{array}{rrcl@{\hspace{1cm}}}
\overline{\sigma} : & \overline{M}  & \longrightarrow & \AO\\ 
    & C & \longmapsto & \sigma(\min_< \pi_1[C \cap(M \times F)]) \\
    &&& \mbox{ where } \pi_1[C \cap(M \times F)] = \{x \in M \mid \exists q \in F, (x,q) \in C\}
\end{array}
\]
To see that $\overline{\sigma}$ is well-defined, note that $\pi_2[C] \cap F \neq \emptyset$ by definition of $\overline{M}$, so $\pi_1[C \cap(M \times F)] \neq \emptyset$,  so $\min_{<} \pi_1[C \cap(M \times F)]$ is well-defined.
\[
\begin{array}{rrcl}
\overline{\mf} : & \overline{M} \times \AT & \longrightarrow & \overline{M}\\
    & (C,\aT) & \longmapsto & \big\{\big(\mf(x,\aT'),\delta(q,(\overline{\sigma}(C),\aT),(\sigma(x),\aT')) \big) \mid (x,q) \in C \wedge \aT' \in \AT\big\}
\end{array}
\]
To argue that $\overline{\mf}$ is well-defined, let $(C,\aT) \in \overline{M} \times \AT$ and let us show that $C' := \overline{\mf}(C,\aT) \in \overline{M}$. Indeed $\pi_2[C'] = \big\{\delta(q,(\overline{\sigma}(C),\aT),(\sigma(x),\aT')) \mid (x,q) \in C \wedge \aT' \in \AT\big\}$. By definition, $\overline{\sigma}(C) = \sigma(y)$ for some $y \in M$, and there exists $q' \in F$ such that $(y,q') \in C$. By setting $(x,q,b') := (y,q',b)$ in the above expression of $\pi_2[C']$, one finds $\delta(q',(\sigma(y),\aT),(\sigma(y),\aT)) \in \pi_2[C']$. On the other hand, $\delta(q',(\sigma(y),\aT),(\sigma(y),\aT)) \in F$ by closedness by adding a suffix, so $\pi_2[C'] \cap F \neq \emptyset$.

Let $\overline{\sO}: \AT^* \to \AO$ be the strategy with memory-aware implementation $(\overline{M},\overline{m}_0,\overline{\sigma},\overline{\mf})$.

{\bf Connecting the local and the global:} 
Let us show by induction on $n$ that for all $n \in \N$, for all $\ol{\aHT} \in \AT^n$,
\begin{align}
\overline{\mf}^+(\ol{\aHT}) = \{(\mf^+(\aHT), \delta^+(\hOS(\overline{\sO},\ol{\aHT}) \| \hOS(\sO,\aHT))) \mid \aHT \in \AT^n\} \label{eq:last-mu-bar}
\end{align}
For the base case, $\overline{\mf}^+(\e) = \{(m_0,q_0)\} = \{(\mf^+(\e), \delta^+(\hOS(\ol{\sO},\e) \| \hOS(\sO,\e)))\}$. For the inductive case,
\begin{align*}
\ol{\mf}^+(\ol{\aHT}\ol{\aT}) & = \ol{\mf}(\ol{\mf}^+(\ol{\aHT}),\ol{\aT}) \mbox{ by definition of }\ol{\mf}^+,\\
    & = \big\{\big(\mf(x,\aT),\delta(q,(\ol{\sO}(\ol{\aHT}),\ol{\aT}),(\sigma(x),\aT)) \big) \mid (x,q) \in \ol{\mf}^+(\ol{\aHT}) \wedge \aT \in \AT\big\}\\
    & \quad\quad \mbox{ by definition of }\ol{\mf} \mbox{ and }\ol{\sO} ,\\
    & = \big\{\big(\mf(\mf^+(\aHT),\aT),\delta(q,(\ol{\sO}(\ol{\aHT})),\ol{\aT}),(\sigma \circ \mf^+(\aHT),\aT)) \big) \mid \aHT \in \AT^n \wedge \aT \in \AT\big\}\\
    & \quad\quad\mbox{ by I.H.}, \mbox{ where }q = \delta^+(\hOS(\overline{\sO},\ol{\aHT}) \| \hOS(\sO,\aHT)) ,\\
    & = \big\{\big(\mf^+(\aHT\aT),\delta(q,(\ol{\sO}(\ol{\aHT})),\ol{\aT}),(\sO(\aHT),\aT)) \big) \mid \aHT \in \AT^n \wedge \aT \in \AT\big\}\\
    & \quad\quad \mbox{ by definition of }\mf^+ \mbox{ and }\sO, \mbox{ where }q \mbox{ is as above},\\
    & = \{(\mf^+(\aHT \aT), \delta^+(\hOS(\overline{\sO},\ol{\aHT} \ol{\aT}) \| \hOS(\sO,\aHT \aT))) \mid \aHT \in \AT^n \wedge \aT \in \AT\}\\
    & \quad\quad \mbox{ by definition of }\delta^+,\, \| \mbox{ and }\hOS, \mbox{ thus completing the induction}.
\end{align*}
A straightforward induction would also show that $\ol{\mf}^+(\ol{\aHT}) \in \ol{M}$ for all $\ol{\aHT} \in \AT^*$. In particular, $\ol{\mf}^+(\ol{\aHT}) \cap (M \times F) \neq \emptyset$ for all $\ol{\aHT} \in \AT^*$.

{\bf Proving that $\overline{\sO}$ is a $\sim$-strategy:} 
First, let us show that for all $n \in \N$ and $\ol{\aHT},\ol{\aHT}' \in \AT^n$,
\begin{align}
\hOS(\overline{\sO}, \ol{\aHT}) \sim \hOS(\overline{\sO}, \ol{\aHT}') \Rightarrow \pi_1[\ol{\mf}^+(\ol{\aHT}) \cap (M \times F)] = \pi_1[\ol{\mf}^+(\ol{\aHT}') \cap (M \times F)]\label{impl:h-pi1-F}
\end{align}
Let us assume that $\hOS(\overline{\sO}, \ol{\aHT}) \sim \hOS(\overline{\sO}, \ol{\aHT}')$. Then, for all $x \in M$,
\begin{align*}
 x \in \pi_1[\ol{\mf}^+(\ol{\aHT}) \cap (M \times F)] &  \Leftrightarrow \exists \aHT \in \AT^n,\, x = \mf^+(\aHT) \wedge \delta^+(\hOS(\overline{\sO},\ol{\aHT}) \| \hOS(\sO,\aHT)) \in F\\
    & \quad\quad \mbox{ by Formula (\ref{eq:last-mu-bar})},\\
    & \Leftrightarrow \exists \aHT \in \AT^n,\, x = \mf^+(\aHT) \wedge \delta^+(\hOS(\overline{\sO},\ol{\aHT}') \| \hOS(\sO,\aHT)) \in F\\
    & \quad\quad \mbox{ since } \hOS(\overline{\sO}, \ol{\aHT}) \sim \hOS(\overline{\sO}, \ol{\aHT}') \mbox{ and } \forall \rho,\rho', \delta^+(\rho \| \rho') \in F \Leftrightarrow \rho \sim \rho',\\
    & \Leftrightarrow x \in \pi_1[\ol{\mf}^+(\ol{\aHT}') \cap (M \times F)] \mbox{ by Formula (\ref{eq:last-mu-bar}) again}.
\end{align*}

Let us now assume that $\hOS(\overline{\sO}, \ol{\aHT}) \sim \hOS(\overline{\sO}, \ol{\aHT}')$ for some $\ol{\aHT},\ol{\aHT}' \in \AT^n$. Then
\begin{align*}
\overline{\sO}(\ol{\aHT}) & = \overline{\sigma} \circ \overline{\mf}^+(\ol{\aHT}) \mbox{ since }(\overline{M}, \overline{m}_0,\overline{\sigma}, \overline{\mf}) \mbox{ is an implementation of }\overline{\sO},\\
  & = \sigma(\min_< \pi_1[\ol{\mf}^+(\ol{\aHT}) \cap M \times F]) \mbox{ by definition of }\overline{\sigma},\\
  & = \sigma(\min_< \pi_1[\ol{\mf}^+(\ol{\aHT}') \cap M \times F]) \mbox{ by Implication (\ref{impl:h-pi1-F}}),\\
  & = \overline{\sO}(\ol{\aHT}') \mbox{ for similar reasons as above, invoked in reverse order}.
\end{align*}

{\bf Proving that $\overline{\sO}$ is a winning strategy (if $\sO$ is):}
let us assume that $\sO$ is winning. Let $\ol{\aHTinf} \in \AT^\omega$. The remainder shows that $\hOS(\ol{\sO},\ol{\aHTinf}) \in W$. For all $n \in \N$, since $\ol{\mf}^+(\ol{\aHTinf}_{\leq n}) \cap (M \times F) \neq \emptyset$, let $\aHT^n \in \AT^n$ such that $\mf^+(\aHT^n) \in \pi_1[\ol{\mf}^+(\ol{\aHT}) \cap M \times F]$. 
So by (\ref{eq:last-mu-bar}), for all $n \in \N$
\begin{align}
\hOS(\ol{\sO},\ol{\aHTinf}_{\leq n}) \sim \hOS(\sO,\aHT^n) \label{eq:h-beta-n}
\end{align}
Consider the tree induced by the $\aHT^n$ (and their prefixes). It is finitely branching since $\AT$ is finite, and it is infinite since each $\aHT^n$ has length $n$. So by Koenig's Lemma it has an infinite branch $\aHTinf$, i.e. for all $n \in \N$ there exists $k_n \in \N$ such that $\aHTinf_{\leq n} = \aHT^{n+k_n}_{\leq n}$. So for all $n \in \N$ there exists $\gamma(n)$ such that $\hOS(\sO,\aHTinf_{\leq n})\gamma(n) = \hOS(\sO,\aHT^{n+|\gamma(n)|})$.

Let $\aHinf := \hOS(\sO,\aHTinf) \in W$ and $\aHinf' := \hOS(\ol{\sO},\ol{\aHTinf})$. So for all $n \in \N$, $\aHinf_{\leq n}\gamma(n) = \hOS(\sO,\aHTinf_{\leq n})\gamma(n) = \hOS(\sO,\aHT^{n+|\gamma(n)|}) \sim \hOS(\ol{\sO},\ol{\aHTinf}_{\leq n +|\gamma(n)|}) = \aHinf'_{\leq n +|\gamma(n)|}$, respectively by definition of $\aHinf$, by definition of $\gamma(n)$,  by (\ref{eq:h-beta-n}), and by definition of $\aHinf'$. 
So by strong $W$-closedness, $\hOS(\ol{\sO},\ol{\aHTinf}) =  \aHinf' \in W$. Since $\ol{\aHTinf}$ was arbitrary, $\ol{\sO}$ is a winning strategy.

\end{proof}

Let us comment on the above proof. First, the memory size is $\mathcal{P}(M_s \times M_{\sim})$ so there is an exponential blow-up compared to $M_s \times F_{\sim}$ from Theorem~\ref{thm:main}.\ref{thm:main1}. Second, whereas the memory-aware part of Theorem~\ref{thm:main}.\ref{thm:main1} was built on its memory-unaware part, ``simply'' by implementing $\sO \circ f_{\sO}$, the above proof has to start from scratch since $\sO'$ has no reason to be implementable with a small memory. Indeed, in the first part of Theorem~\ref{thm:main}.\ref{thm:main2} the choice of one representative for each opponent-history requires backtracking that relies on a new $f_\sO$ which has no reason to be computationally easy to handle. A new, suitable strategy is thus defined by combining elements of the implementation of $\sO$ and of $\sim$. Third, finiteness of $\AT$ was assumed in Proposition~\ref{prop:finite-wawc} but not in its improvement in Theorem~\ref{thm:main}.\ref{thm:main2}. The proof was strengthened by replacing Koenig's Lemma with an argument involving lexicographic well-orders. I failed to do the same for the memory-aware part of Theorem~\ref{thm:main}.\ref{thm:main2}.

\begin{proof}[Proof of Corollary~\ref{cor:main}.\ref{cor:main2}]
Let $\sO$ be a \pO\/ winning strategy. Let $\sim$ be an equivalence relation over $(\AO \times \AT)^*$ defined by $\aH \sim \aH'$ if $\tr^{++}(\aH) = \tr^{++}(\aH')$ and $\delta^{+}(\aH) = \delta^{+}(\aH')$. This $\sim$ is time-aware since $\tr^{++}(\aH) = \tr^{++}(\aH')$ implies $|\aH| = |\aH'|$. This $\sim$ is closed by adding a suffix: let us assume $\aH \sim \aH'$, and let $\aH'' \in (\AO \times \AT)^*$; it is straightforward to show by induction on $\aH''$ that $\aH \aH''\sim \aH' \aH''$. 

Let us show that $\sim$ is strongly $W$-closed. Let $\aHinf, \aHinf' \in (\AO \times \AT)^\omega$ be such that for all $n \in \N$ there exists $\gamma \in (\AO \times \AT)^*$ such that $\aHinf_{\leq n}\gamma \sim \aHinf'_{\leq n+|\gamma|}$, implying $\tr^{++}(\aHinf_{\leq n}\gamma) = \tr^{++}(\aHinf'_{\leq n+k})$, and $\tr^{++}( \aHinf_{\leq n}) = \tr^{++}( \aHinf'_{\leq n})$. Therefore $\tr^{\infty}(\aHinf) = \tr^{\infty}(\aHinf')$, and $\aHinf \in (\tr^{\infty})^{-1}[W]$ iff $\aHinf' \in (\tr^{\infty})^{-1}[W]$.

Recall that, by definition, the strategies in $G' = \langle \AO, \AT, (\tr^{\infty})^{-1}[W]\rangle$ are exactly the strategies in $G$, and that a strategy is winning in $G$ iff it is winning in $G'$. So, by Theorem~\ref{thm:main}.\ref{thm:main2} applied to the game $G'$, there exists a winning $\sim$-strategy in $G$.

For the memory, it suffices to show that $\sim$ is $2$-tape recognized by an automaton with $|Q|^2+1$ states. Let $q_{-1} \notin Q$ be a failure state and let us define the automaton $\langle (\AO \times \AT)^2,Q^2 \cup \{q_{-1}\}, (q_0,q_0), \{(q,q) \mid q \in Q\},\delta_{\sim}\rangle$, where
\begin{align*}
 \delta_{\sim}((q,q'),(\aO,\aT),(\aO',\aT')) & := (\delta(q,\aO,\aT),\delta(q',\aO',\aT'))\mbox{ if } \tr(q,\aO,\aT) = \tr(q',\aO',\aT')\\
 \delta_{\sim}((q,q'),(\aO,\aT),(\aO',\aT')) & := q_{-1} \mbox{ if } \tr(q,\aO,\aT) \neq \tr(q',\aO',\aT')\\
\delta_{\sim}(q_{-1},(\aO,\aT),(\aO',\aT')) & := q_{-1}
\end{align*}
\end{proof}

\begin{proof}[Proof of Corollary~\ref{cor:shuffle}]
The game $G$ is a Muller game, so, by \cite{GH82}, \pO\/ who has a winning strategy also has a finite-memory one. Let a constraint $\sim$ be defined by $\aH \sim \aH'$ if the multisets induced by $\tr^{++}(\aH)$ and $\tr^{++}(\aH')$ are equal and $\delta^+(\aH) = \delta^+(\aH')$. Let us prove useful properties of $\sim$ below. 
\begin{itemize}
 \item Time awareness: if $\aH \sim \aH'$, then $|\tr^{++}(\aH)| = |\tr^{++}(\aH')|$, and $|\aH| = |\aH'|$. 
 \item Closedness by adding a suffix: if $\aH \sim \aH'$, on the one hand, $\tr^{++}(\aH)$ and $\tr^{++}(\aH')$ induce the same multisets, and so do $\tr^{++}(\aH \aH'')$ and $\tr^{++}(\aH' \aH'')$; on the other hand, $\delta^+(\aH) = \delta^+(\aH')$, so $\delta^+(\aH\aH'') = \delta^+(\aH'\aH'')$; therefore $\aH \aH''\sim \aH'\aH''$.
 
 \item Strong $W$-closedness: Let $\aHinf \in W$ and $\aHinf' \in (\AO \times \AT)^\omega$ be such that $(\forall n \in \N, \exists \gamma \in (\AO \times \AT)^*, \aHinf_{\leq n}\gamma \sim \aHinf'_{\leq n+|\gamma|})$, so for all $n \in \N$, there exists $k \in \N$, such that every color occurring in $\tr^{++}(\aHinf_{\leq n})$ occurs at least as many times in $\tr^{++}(\aHinf'_{\leq k})$. Since $\aHinf \in W$, there exists some $C_i$ such that all the elements of $C_i$ occur infinitely often in $\tr^{\infty}(\aHinf)$, and also in $\tr^{\infty}(\aHinf')$ by the above remark, so $\aHinf' \in W$.
\end{itemize}
Therefore, by Theorem~\ref{thm:main}.\ref{thm:main2}, \pO\/ has a winning finite-memory $\sim$-strategy.
\end{proof}

 Example~\ref{exa:uni-mem} below is quite informative. It shows a game with a winning strategy that is smaller than all winning $\sim$-strategies. More specifically, a winning strategy can be implemented with memory $Q$, whereas implementations of winning $\sim$-strategies all require memory size greater than $Q$. The game is morally turn-based and finite but expressed in the general formalism.
\begin{example}\label{exa:uni-mem}
Let $\AO = \AT = \{0,1\}$. Let \pO\/ win if she plays in the third round what \pT\/ had played in the first round, or if \pT\/ plays $01$ or $10$ in the first two rounds. Let $\sim$ be the smallest constraint generated by $(\_,0)(\_,1) \sim (\_,1)(\_,0)$ and closed by adding a suffix. The underscore above means anything in $\{0,1\}$, quantified universally, i.e. $(x,0)(y,1) \sim (z,1)(t,0)$ for all $x,y,z,t \in \{0,1\}$. This $\sim$ is also time-aware and strongly $W$-closed. One can implement a winning strategy using the left-hand side automaton in Figure~\ref{fig:asws}, since it suffices to remember whether \pT\/ played $0$ or $1$ in the first round. But intuitively, implementing a winning $\sim$-strategy requires to remember the two actions of \pT\/ to detect $01$ and $10$ and react similarly for both, as in the right-hand side of Figure~\ref{fig:asws}.
 \begin{figure}
  \centering
\begin{tikzpicture}[shorten >=1pt,node distance=2cm,on grid,auto]
\node[state] (q1) {$q_1$};
\node[right of = q1, state, initial above] (q0) {$q_0$};
\node[state, right of = q0] (q2) {$q_2$};
\draw (q0) edge[->,above] node{$\_,0$} (q1)
(q1) edge[loop left] (q1)
(q0) edge[->] node{$\_,1$}(q2)
(q2) edge[loop right] (q2);
\end{tikzpicture}
\begin{tikzpicture}[shorten >=1pt,node distance=2cm,on grid,auto]
\node[state, initial] (q0) {$q_0$};
\node[right of = q0, state] (q2) {$q_2$};
\node[state, below of = q0] (q1) {$q_1$};
\node[right of = q1, state] (q3) {$q_3$};
\draw (q0) edge[->,left] node{$\_,0$} (q1)
(q1) edge[loop left] node{$\_,0$} (q1)
(q0) edge[->] node{$\_,1$}(q2)
(q2) edge[loop right] node{$\_,1$} (q2)
(q1) edge[->] node{$\_,1$}(q3)
(q2) edge[->] node{$\_,0$}(q3)
(q3) edge[loop right] (q3);
\end{tikzpicture}
  \caption{Absence of small winning $\sim$-strategies}\label{fig:asws}
\end{figure}
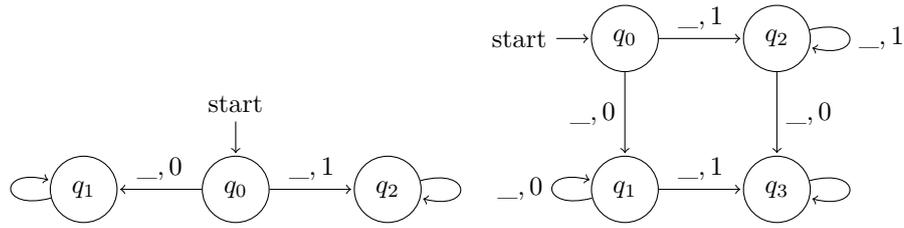
\end{example}

\section{Properties of used and alternative concepts}\label{sect:cbusc}

This section first compares the two notions of recognizability discussed in the introduction: $2$-tape recognizability, used in this article, and $1$-tape recognizability. Second, it compares strong $W$-closedness, weak $W$-closedness, and perfect recall in several ways. Third, it proves algebraic properties of the five constraint predicates. 

Proposition~\ref{prop:regular-eq-rel} below shows that, if considering time-aware equivalence relations, the concept of $2$-tape recognizability is more general than the concept of $1$-tape regularity. So this article uses the former to obtain stronger results. The {\bf time-aware restriction} $\sim_{ta}$ of an equivalence relation $\sim$ is defined by $\aH \sim_{ta} \aH'$ iff $\aH \sim \aH'$ and $|\aH| = |\aH'|$.

\begin{proposition}\label{prop:regular-eq-rel}
\begin{enumerate}
 \item\label{prop:regular-eq-rel1} If an equivalence relation is $1$-tape recognizable using memory $Q$ (to build an automaton recognizing it), its time-aware restriction is $2$-tape regular using memory $Q^2$.

 \item\label{prop:regular-eq-rel2} Every time-aware equivalence relation over $\Sigma^*$ is $2$-tape recognizable using memory $\Sigma^*$.

 \item\label{prop:regular-eq-rel3} Every $1$-tape regular equivalence relation is closed by adding a suffix.
 
 
 \item\label{prop:regular-eq-rel4} There exists a $2$-tape regular equivalence relation that is closed by adding a suffix, but is not the time-aware restriction of any $1$-tape regular equivalence relation.
\end{enumerate}

\begin{proof}
\begin{enumerate}
 \item Let $\sim$ be a $1$-tape regular equivalence relation, and let $(\Sigma,Q,q_0,\delta)$ be a corresponding automaton. Let $\mathcal{A} := (\Sigma^2,Q^2,(q_0,q_0),F,\delta_2)$, where $F := \{(q,q) \mid q \in Q\}$ and $\delta_2(q_1,q_2, a_1, a_2) := (\delta(q_1, a_1),\delta(q_2, a_2))$ for all $q_1,q_2 \in Q$ and $a_1,a_2 \in \Sigma$. It is easy to check that for all $u,v \in \Sigma^*$ of equal length, $\delta_2^+(u \| v) \in F$ iff $\delta^+(u) = \delta^+(v)$ iff $u \sim_{ta} v$. So $\sim_{ta}$ is $2$-tape recognizable.
 
 \item Let $\sim$ be a time-aware relation over $\Sigma^*$. The automaton $(\Sigma^2,Q,q_0,F,\delta)$ recognizes $\sim$, where $Q := \{(u,v) \in (\Sigma^*)^2 \mid |u| = |v|\}$, $q_0 := (\e,\e)$, $F := \sim$, and $\delta(u,v,a,b) := (ua,vb)$. 
 
 \item Reading the same word after reaching the same state leads to the same state.
 
 \item Let $\sim$ be defined by $0^n \sim 0^n$ and $0^n1u \sim 0^n1v$ for all $n \in \N$ and $u,v \in \{0,1\}^*$ such that $|u| = |v|$. It is closed by adding a suffix. The automaton depicted in Figure~\ref{fig:pairwise-not-single} witnesses that $\sim$ is $2$-tape recognizable. (Formally, the automaton is $(\{0,1\}^2,\{q_0,q_1,q_2\},q_0,\{q_0,q_1\},\delta)$, where $\delta(q_0,0,0) = q_0$, $\delta(q_0,1,1) = q_1$, $\delta(q_0,0,1) = \delta(q_0,1,0) = q_2$, and $\delta(q_i,\_,\_) = q_i$ for $i \in \{1,2\}$.) The words of length $n$ are partitioned into $n+1$ equivalence classes, so the relation $\sim$ is not the time-aware restriction of any $1$-tape regular relation using finite memory, since the latter would have finitely many equivalence classes.
 \begin{figure}
  \centering
\begin{tikzpicture}[shorten >=1pt,node distance=2cm,on grid,auto]
\node[state, accepting] (q1) {$q_1$};
\node[right of = q1, state, initial above, accepting] (q0) {$q_0$};
\node[state, right of = q0] (q2) {$q_2$};
\draw (q0) edge[loop below] node{$0,0$} (q0)
(q0) edge[->,above] node{$1,1$} (q1)
(q1) edge[loop left] (q1)
(q0) edge[->] node{$0,1$} node[swap]{$1,0$}(q2)
(q2) edge[loop right] (q2);
\end{tikzpicture}
  \caption{A $2$-tape recognizable relation not coming from a $1$-tape regular one}\label{fig:pairwise-not-single}
\end{figure}
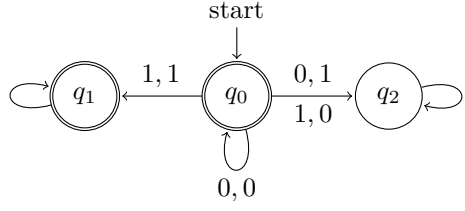

\end{enumerate}

\end{proof}
\end{proposition}

For the remainder of this section, let $\langle \AO,\AT,W \rangle$ be a game and $\sim$ be a constraint. In addition to comparing strong and weak $W$-closedness, let us define a notion lying in between, as will be proved, and used by a tightness result in Section~\ref{sect:tight}. The constraint $\sim$ is said to be $W$-closed, if $\forall \aHinf,\aHinf' \in (\AO \times \AT)^\omega,\,(\forall n \in \N, \exists k \in \N, \aHinf_{\leq n+k} \sim \aHinf'_{\leq n+k}) \Rightarrow (\aHinf \in W \Leftrightarrow \aHinf' \in W)$.

Proposition~\ref{prop:la-wlog} below shows that the assumption of time awareness is done without loss of generality in the following (weak) sense: if $\sim$ satisfies all the requirements of Theorem~\ref{thm:main} but time awareness, its time-aware restriction satisfies all requirements. The proofs are straightforward.
\begin{proposition}\label{prop:la-wlog}
\begin{enumerate}
 \item If $\sim$ is closed by adding a suffix, so is its time-aware restriction $\sim_{ta}$.
 
 \item If $\sim$ is perfectly recalling, so is its time-aware restriction $\sim_{ta}$.
 
 \item If $\sim$ is (strongly, weakly) $W$-closed, so is its time-aware restriction $\sim_{ta}$.
\end{enumerate}
\end{proposition}

Proposition~\ref{prop:closednesses-impl} below justifies the attributes strong and weak for $W$-closedness.
\begin{proposition}\label{prop:closednesses-impl}
Strong $W$-closedness implies $W$-closedness, which implies weak $W$-closedness. Furthermore, the converse implications are false even under assumption of time awareness and closedness by adding a suffix.

\begin{proof}
It is straightforward to show that $W$-closedness implies weak $W$-closedness.

Consider some $\sim$ that is strongly $W$-closed, and let $\aHinf,\aHinf' \in (\AO \times \AT)^\omega$ be such that for all $n \in \N$ there is $k \in \N$ such that $\aHinf_{\leq n+k} \sim \aHinf'_{\leq n+k}$. Using such $n$ and $k$, let $\gamma(n)$ be such that $\aHinf_{\leq n}\gamma(n) = \aHinf_{\leq n+k}$, so $\aHinf_{\leq n}\gamma(n) \sim \aHinf'_{\leq n+|\gamma(n)|}$. So $\aHinf \in W \Rightarrow \aHinf' \in W$ by strong $W$-closedness. By swapping $\aHinf$ and $\aHinf'$, one obtains $\aHinf' \in W \Rightarrow \aHinf \in W$ and the equivalence, so $\sim$ is $W$-closed.

To show that weak $W$-closedness does not imply $W$-closedness, let $\AO = \{0,1\}$ and $\AT = \{0\}$ and let $\aH \sim \aH'$ iff $|\aH| = |\aH'|$ and $\sum_{i = 1}^{|\aH|}\pi_1(\aH_i) = \sum_{i = 1}^{|\aH|}\pi_1(\aH'_i)$. This $\sim$ is time-aware by definition, and it is clearly closed by adding a suffix. Furthermore, $\sim$ is weakly $W$-closed regardless of $W$ since $(\forall n \in \N,\aHinf_{\leq n} \sim \aHinf'_{\leq n})$ implies $\aHinf = \aHinf'$. This is proved by showing $\forall n \in \N,\aHinf_{\leq n} = \aHinf'_{\leq n}$ by induction on $n$. Now let $W = \{\aHinf\}$ where $\aHinf := ((0,0)(1,0))^\omega$, and let $\aHinf' := ((0,0)(0,0)(1,0)(1,0))^\omega \notin W$. Yet $\aHinf_{\leq 2n} \sim \aHinf'_{\leq 2n}$ for all $n \in \N$, so $\sim$ is not $W$-closed.   

To show that $W$-closedness does not imply strong $W$-closedness, let $\AO = \{0,1\}$ and $\AT = \{0\}$ and let $\aH \sim \aH'$ iff $|\aH| = |\aH'|$ and either $\aH = \aH' = (0,0)^{|\aH|}$ or they are both different from  $(0,0)^{|\aH|}$. This $\sim$ is time-aware by definition, and it is also clearly closed by adding a suffix. Let $W := \{(0,0)^\omega\}$. Then, $\sim$ is $W$-closed: let $\aHinf$ and $\aHinf'$ be such that for all $n \in \N$ there exists $k \in \N$ such that $\aHinf_{\leq n+k} \sim \aHinf'_{\leq n+k}$. Then either the two of them are in $(0,0)^*$, or $1$ occurs in each of them, so $\aHinf \in W \Leftrightarrow \aHinf' \in W$. However, $\sim$ is not strongly $W$-closed: let $\aHinf := (0,0)^\omega$ and $\aHinf' := (1,0)^\omega$, and $\gamma := (1,0)$. Thus, for all $n \in \N$ we have $\aHinf_{\leq n}\gamma = (0,0)^{n}(1,0) \sim (1,0)^{n+1} = \aHinf'_{\leq n+|\gamma|}$, yet $\aHinf \in W$ and $\aHinf' \notin W$.
\end{proof}

\end{proposition}

Proposition~\ref{prop:closednesses-collapse} below show that the implications of Proposition~\ref{prop:closednesses-impl} are equivalences under perfect recall. This already suggests that perfect recall is a strong assumption.
\begin{proposition}\label{prop:closednesses-collapse}
A perfectly recalling, weakly $W$-closed constraint is also strongly $W$-closed.

\begin{proof}
Let $\sim$ be perfectly recalling and weakly $W$-closed. So $\aH \sim \aH'\,\Rightarrow\, \forall n \leq |\aH|,\, \aH_{\leq n} \sim \aH_{\leq n}$ and $(\forall n \in \N, \aHinf_{\leq n} \sim \aHinf'_{\leq n}) \Rightarrow (\aHinf \in W \Leftrightarrow \aHinf' \in W)$.

Let $\aHinf,\aHinf' \in (\AO \times \AT)^\omega$ be such that $\forall n \in \N,\exists \gamma \in (\AO \times \AT)^\omega, \aHinf_{\leq n} \gamma \sim \aHinf'_{\leq n+|\gamma|}$. So by perfect recall, $\aHinf_{\leq n}\sim \aHinf'_{\leq n}$ for all $n \in \N$. Therefore $\aHinf\sim \aHinf'$ by weak $W$-closedness.
\end{proof}
\end{proposition}

\begin{proof}[Proof of Lemma~\ref{lem:constraint-algebra}]
 \begin{enumerate}
  \item Let $\aH (\cap_{i \in I}\sim_i) \aH'$, so $\aH \sim_i \aH'$ for all $i \in I$. In particular, $\aH \sim_i \aH'$ for some $i \in I$, so $|\aH| = |\aH'|$ if all $\sim_i$ are time aware. Let $\aH'' \in (\AO \times \AT)*$, so $\aH \aH'' \sim_i \aH' \aH''$ for all $i \in I$ if all the $\sim_i$ are closed by adding a suffix, implying $\aH \aH''(\cap_{i \in I}\sim_i) \aH'\aH''$. If $|\aH| = |\aH'|$ then for all $i \in I$ and all $n \leq |\aH|$ we have $\aH_{\leq n} \sim_i \aH_{\leq n}$. Hence, $\aH_{\leq n} \sim \aH_{\leq n}$ for all $n \leq |\aH|$.  
   
  \item Let $\aHinf,\aHinf' \in (\AO \times \AT)^\omega$ be such that $ \aHinf_{\leq n} (\cap_{i \in I}\sim_i)\aHinf'_{\leq n}$ for all $n \in \N$ and $\aHinf \in W := \cap_{i \in I}W_i$. So for all $i \in I$ we have $ \aHinf_{\leq n} \sim_i \aHinf'_{\leq n}$ for all $n \in \N$, so $\aHinf' \in W_i$ by weak $W_i$-closedness and since $\aHinf \in W_i$. So $\aHinf' \in W$.

\item Let $\aHinf,\aHinf' \in (\AO \times \AT)^\omega$ be such that for all $n \in \N$, there exists $\gamma \in (\AO \times \AT)^*$ such that $ \aHinf_{\leq n}\gamma (\cap_{i \in I}\sim_i) \aHinf'_{\leq n+|\gamma|}$, and such that $\aHinf \in W := \cap_{i \in I}W_i$. So for all $i \in I$, for all $n \in \N$, there exists $\gamma \in (\AO \times \AT)^*$ (the same $\gamma$) such that $ \aHinf_{\leq n}\gamma \sim_i \aHinf'_{\leq n+|\gamma|}$, so $\aHinf' \in W_i$ by strong $W_i$-closedness and since $\aHinf \in W_i$. So $\aHinf' \in W$.

  \item Let $\aHinf,\aHinf' \in (\AO \times \AT)^\omega$ be such that $ \aHinf_{\leq n}\sim \aHinf'_{\leq n}$ for all $n \in \N$, and such that $\aHinf \in W := \cup_{i \in I}W_i$. Let $j \in I$ be such that $\aHinf \in W_j$. By weak $W_j$-closedness, $\aHinf' \in W_j \subseteq W$.
  
  \item Let $\aHinf,\aHinf' \in (\AO \times \AT)^\omega$ be such that for all $n \in \N$, there exists $\gamma \in (\AO \times \AT)^*$ such that $ \aHinf_{\leq n}\gamma \sim \aHinf'_{\leq n+|\gamma|}$, and such that $\aHinf \in W := \cup_{i \in I}W_i$. Let $j \in I$ be such that $\aHinf \in W_j$. By strong $W_j$-closedness, $\aHinf' \in W_j \subseteq W$.
  
 \end{enumerate}
\end{proof}

\section{Tightness results}\label{sect:tight}

This section shows tightness results for Theorem~\ref{thm:main}. Dropping perfect recall falsifies Theorem~\ref{thm:main}.\ref{thm:main1}. Dropping either time awareness or closedness by adding a suffix
falsifies Theorem~\ref{thm:main}.\ref{thm:main2}. Despite this relative tightness, the end of the section shows well-known examples that are not captured by Theorem~\ref{thm:main}.\ref{thm:main2}.

Given a game with a \pO\/ winning strategy, given a constraint $\sim$, if there are no \pO\/ winning $\sim$-strategies, then $\sim$ is said to be harmful. Otherwise it is said to be harmless.

Proposition~\ref{prop:sasc-tp} below shows that perfect recall cannot be simply dropped in Theorem~\ref{thm:main}.\ref{thm:main1}.

\begin{proposition}\label{prop:sasc-tp}
There exist a game $\langle \{0,1\},\{0,1\}, W\rangle$ and a constraint that is time-aware, closed by adding a suffix, weakly $W$-closed, and yet harmful.

  \begin{proof}
Let $W$ be defined as follows: $\aHinf \in W$ iff $\pi_1(\aHinf_{1}) = \pi_1(\aHinf_{2}) = 1 - \pi_2(\aHinf_{0})$. (This is morally a three-round turn-based game.) Let us define $\sim$ as follows: for all $\aH,\aH' \in (\{0,1\}^2)^*$, let $\aH \sim \aH'$ iff $0 < |\aH| = |\aH'|$ and $\pi_2(\aH_0) + \sum_{i = 1}^{|\aH|-1}\pi_1(\aH_i) = \pi_2(\aH'_0) + \sum_{i = 1}^{|\aH|-1}\pi_1(\aH'_i)$. So, $\sim$ is time-aware by definition, it is clearly closed by adding a suffix, and also weakly $W$-closed.

$(0,0)(1,1) \sim (1,1)(0,0)$ but after $(0,0)(1,1)$ \pO\/ must play $1$ to win, and after $(1,1)(0,0)$ she must play $0$ to win. So there is no winning $\sim$-strategy.
  \end{proof}
\end{proposition}

Proposition~\ref{prop:wawc-la} below shows that time awareness cannot be simply dropped in Theorem~\ref{thm:main}.\ref{thm:main2}.

\begin{proposition}\label{prop:wawc-la}
There exist a game $\langle \{0,1\},\{0\}, W\rangle$ and a constraint that is closed by adding a suffix, strongly $W$-closed, and yet harmful.
\begin{proof}
Let $W$ be the set of runs such that \pO\/ plays both $0$ and $1$ infinitely often, and let two histories $\aH, \aH' \in (\{0,1\} \times \{0\})^*$ be $\sim$-equivalent if \pO\/ has played $1$ the same number of times. (Morally, it is a one-player game.)

This $\sim$ is clearly closed by adding a suffix. Let us argue that it is also strongly $W$-closed. Let $\aHinf \in W$ and $\aHinf' \in (\{0,1\} \times \{0\})^\omega$ be such that $\forall n \in \N,\exists \gamma \in (\{0,1\} \times \{0\})^*,\, \aHinf_{\leq n}\gamma \sim \aHinf'_{\leq n+|\gamma|}$. This implies that for all $n \in \N$, $\aHinf'$ involves at least as many $0$'s and as many $1$'s as $\aHinf_{\leq n}$. So $\aHinf'$ involves infinitely many $0$'s and $1$'s, so it is in $W$.

\pO\/ can win by playing $0$ and $1$ alternately, but any $\sim$-strategy always prescribes $0$ after the first time it prescribes $0$, since the number of occurrences of $1$ will then remain the same. So \pO\/ plays $1$ either always or only finitely many times, which is not winning.
\end{proof}
\end{proposition}

Proposition~\ref{prop:wawc-as} below shows that the assumption of closedness by adding a suffix cannot be simply dropped in Theorem~\ref{thm:main}.\ref{thm:main2}.

\begin{proposition}\label{prop:wawc-as}
There exist a game $\langle \{0,1\},\{0,1\}, W\rangle$ and a constraint that is time-aware, strongly $W$-closed, and yet harmful.
\begin{proof}
Let $W := (\_,0)(0,\_)(\{0,1\}^2)^\omega \cup (\_,1)(1,\_)(\{0,1\}^2)^\omega$, where the underscore means either $0$ or $1$. Said otherwise, for \pO\/ to win, her second action should imitate \pT\/'s first action. This is morally a two-round turn-based game. Let $\sim$ be defined by $(\_,\_) \sim (0,0)$ and $(\_,0)(0,\_)u \sim (\_,1)(1,\_)v$ and $(\_,1)(0,\_)u \sim (\_,0)(1,\_)v$ for all $u,v \in (\{0,1\}^2)^*$ such that $|u|=|v|$.

The definition of $\sim$ ensures time awareness. Let us argue that it is also strongly $W$-closed. Let $\aHinf \in W$, and let $\aHinf' \in (\{0,1\}^2)^\omega$ be such that $\forall n \in \N,\exists \gamma \in (\{0,1\}^2)^*, \aHinf_{\leq n} \gamma \sim \aHinf'_{\leq n+|\gamma|}$. For $n := 2$ this assumption provides some $\gamma$ such that $\aHinf_{\leq 2} \gamma \sim \aHinf'_{\leq 2+|\gamma|}$. Since $\aHinf \in W$, $\aHinf_{\leq 2}$ matches $(\_,0)(0,\_)$ or $(\_,1)(1,\_)$, and so does $\aHinf'_{\leq 2}$ by definition of $\sim$, so $\aHinf' \in W$.

Clearly \pO\/ has a winning strategy for this game, but every $\sim$-strategy prescribes the same action after $(\_,0)$ and $(\_,1)$, which is not winning.
\end{proof}
\end{proposition}

\paragraph*{Limitations and opportunity for meaningful generalizations}

Despite its relative tightness, Theorem~\ref{thm:main} does not imply all known results that can be seen as instances of strategy uniformization problems, so there is room for meaningful generalizations. E.g., due to time awareness requirement, Theorem~\ref{thm:main} does not imply positional determinacy of parity games~\cite{EJ91,Mostowski91}, where two histories are equivalent if they lead to the same state. Nor does it imply countable compactness of first-order logic~\cite{Goedel30}, which is also an instance of a uniformization problem: Let $(\varphi_n)_{n \in \N}$ be first-order formulas, and define a turn-based game: \emph{Spoiler} plays only at the first round by choosing $m \in \N$. Then \emph{Verifier} gradually builds a countable structure over the signature of $(\varphi_n)_{n \in \N}$. More specifically, at every round she either chooses the value of a variable, or the output value of a function at a given input value, or the Boolean value of a relation for a given pair of values. Only countably many pieces of information are needed to define the structure, and one can fix an order (independent of $m$) in which they are provided. Verifier wins if the structure she has defined is a model of $\land_{0 \leq k \leq m} \varphi_{k}$. Let all histories of equal length be $\sim$-equivalent. Compactness says that if each $\land_{0 \leq k \leq m} \varphi_{k}$ has a model, so does $\land_{0 \leq k} \varphi_{k}$. Said otherwise, if Verifier has a winning strategy, she has a winning $\sim$-strategy, i.e. independent of Spoiler's first move. This $\sim$ satisfies all the conditions of Theorem~\ref {thm:main}.\ref{thm:main1} but weak $W$-closedness: the premise $(\forall n \in \N, \aHinf_{\leq n} \sim \aHinf'_{\leq n})$ holds by universality, but the conclusion $(\aHinf \in W \Leftrightarrow \aHinf' \in W)$ is false since a model for $\land_{0 \leq k \leq m} \varphi_{k}$ need not be a model for $\land_{0 \leq k \leq m+1} \varphi_{k}$.

\section{Maximal harmless contraints}\label{sect:mhc}

This section studies the existence of maximal harmless constraints. Theorem~\ref{thm:finite-coarsest-sim} provides a basic sufficient condition for existence to hold, and examples show its relative tightness.

Remark: the problem of strategy maximal uniformization was advertised in the introduction as relevant to security (to maximize information concealment), but it is also relevant, e.g., to minimize the memory size of a finite-memory winning strategy. 

Proposition~\ref{prop:no-coarsest-sim} below dashes the hope for uniqueness of maximal harmless constraints even in finite turn-based games. (Below, openness and closedness both refer to the topology induced by the distance from Section~\ref{sect:main-def-res}, and clopen means closed and open.)

\begin{proposition}\label{prop:no-coarsest-sim}
There exists a game $\langle\{0,1\}, \{0,1,2\}, W \rangle$ with clopen $W$ (i.e. the game is morally finite), a \pO\/ winning strategy, but no unique maximal harmless constraint.
\begin{proof}
Consider the game $\langle\{0,1\}, \{0,1,2\}, W \rangle$ where a run is in $W$ iff it satisfies the following: if \pT\/ picks $i \in \{0,1\}$ in the first round, \pO\/ picks the same $i$ in the second round. $W$ is a clopen set, and the game is morally a two-round turn-based game. The issue comes from \pO\/'s winning inconditionnally if \pT\/ has picked $2$ in the first round.

Formally, let $\sim_0$ and $\sim_1$ consist of two equivalence classes each: for each $i \in \{0,1\}$, one class of $\sim_i$ is $(\_,i) = \{(0,i),(1,i)\}$. For each $i \in \{0,1\}$, the \pO\/ strategy consisting in picking $i$ after $(\_,i)$, and in picking $1-i$ otherwise is a \pO\/ winning $\sim_i$-strategy. However, the only equivalence relation including $\sim_0$ and $\sim_1$ is the universal relation $\sim$, but there is no \pO\/ winning $\sim$-strategy, since \pO\/ must pick $0$ after $(\_,0)$ and $1$ after $(\_,1)$.
\end{proof}
\end{proposition}

\begin{theorem}\label{thm:finite-coarsest-sim}
Consider a game $\langle \AO, \AT, W \rangle$. There is (at least) one largest harmless constraint with finitely many equivalence classes, if either of the following holds:
\begin{enumerate}
\item\label{thm:finite-coarsest-sim1} there is a \pO\/ winning $\sim$-strategy for a $\sim$ with finitely many equivalence classes,

\item\label{thm:finite-coarsest-sim2} there is a \pO\/ winning strategy and $\AO$ is finite,

\item\label{thm:finite-coarsest-sim3} there is a \pO\/ winning strategy, $\AT$ is finite, and $W$ is open. 
\end{enumerate}
\begin{proof}
\begin{enumerate}
\item Keep merging any two equivalence classes of $\sim$, thus producing larger and larger $\sim'$ relations, as long as there are \pO\/ winning $\sim'$-strategies. The process stops by finiteness, and yields a sought $\sim'$.

\item Let $\sO$ be a \pO\/ winning strategy, and let $\aH \sim \aH'$ iff $\sO(\aH) = \sO(\aH')$. The equivalence relation $\sim$ has finitely many classes since $\AO$ is finite, and $\sO$ is a \pO\/ winning $\sim$-strategy. One concludes by Theorem~\ref{thm:finite-coarsest-sim},\ref{thm:finite-coarsest-sim1}. 

\item Let $\sO$ be a \pO\/ winning strategy. Consider the set $T$ of the $\aHT \in \AT^*$ such that $\hOS(\sO,\aHT) \cdot (\AO \times \AT)^\omega \not\subseteq W$, and let us show that $T$ is a tree. For all $\aHT,\aHT' \in \AT^*$, if the prefix relation $\aHT \sqsubseteq \aHT'$ holds, $\hOS(\sO,\aHT) \sqsubseteq\hOS(\sO,\aHT')$, in which case $\hOS(\sO,\aHT') \cdot (\AO \times \AT)^\omega \subseteq \hOS(\sO,\aHT) \cdot (\AO \times \AT)^\omega$, and subsequently $\hOS(\sO,\aHT') \cdot (\AO \times \AT)^\omega \not\subseteq W$ implies $\hOS(\sO,\aHT) \cdot (\AO \times \AT)^\omega \not\subseteq W$. This shows that $T$ is a tree. Since $\sO$ is winning, $\hOS(\sO,\AT^\omega) \subseteq W$, so since $W$ is open, $T$ as no infinite path. Also, $T$ is finite-branching since $A_2$ is finite, so by Koenig's Lemma $T$ is finite.

Let the strategy $\sO'$ be defined by $\sO'(\aH) := \sO(\aH)$ for all $\aH \in \hOS(\sO,T)$, and $\sO'(\aH) := \sO(\e)$ otherwise. So $\sO'$ is a \pO\/ winning strategy: it behaves like $\sO$ until the run ``leaves'' $T$ from its leaves, but then \pO\/ wins no matter what, by definition of $T$. Let $\sim$ be defined by $\aH \sim \aH'$ iff $\sO'(\aH) = \sO'(\aH')$. It has finitely many equivalence classes (exactly $|\sO \circ \hOS(\sO,T)|$), so by Theorem~\ref{thm:finite-coarsest-sim}.\ref{thm:finite-coarsest-sim1} there is (at least) one largest harmless constraint.
\end{enumerate}
\end{proof}
\end{theorem}

Proposition~\ref{prop:no-coarsest-sim2} below shows the tightness of Theorem~\ref{thm:finite-coarsest-sim}.\ref{thm:finite-coarsest-sim2}.

\begin{proposition}\label{prop:no-coarsest-sim2}
There exists a game $\langle \N, \N, W \rangle$ with clopen $W$ and a \pO\/ winning strategy, but no maximal harmless constraint.
\begin{proof}
Consider the game $\langle \N, \N, W \rangle$ where \pO\/ wins iff  her pick in the second round is greater than \pT\/'s pick in the first round. $W$ is clopen, and the game amounts to a two-round turn-based game. The \pO\/ strategy consisting in picking $n+1$, where $n$ was \pT\/'s pick, is clearly winning. Let a constraint $\sim$ be such that there exists a winning $\sim$-strategy $\sO$. 

Since $\sO$ is winning, $n < \sO(n)$ for all $n \in \N$. In particular $\sO(0) < \sO(\sO(0))$, so the $\sim$-classes of $(\sO(\e),0)$ and $(\sO(\e),\sO(0))$ are distinct since $\sO$ is a $\sim$-strategy. Let $\sim'$ be defined by merging these two classes, so that $\sim'$ is strictly greater than $\sim$. Let us derive a new \pO\/ strategy $\sO'$ as follows: $\sO'(\aH) := \sO(\sO(0))$ if $\aH \sim (\sO(\e),0)$, and $\sO'(\aH) := \sO(\aH)$ otherwise. Thus, $\sO'$ is a \pO\/ winning $\sim'$-strategy.
\end{proof}
\end{proposition}

The action sets in Proposition~\ref{prop:no-coarsest-sim2} are infinite. Proposition~\ref{prop:no-coarsest-sim3} below also shows that there may not be any maximal harmless constraint, but with a finite action set for \pT\/. This is done at the cost of a slight set-theoretic complexification of the winning condition, which is still closed but no longer open. The benefit is that Proposition~\ref{prop:no-coarsest-sim3} shows the tightness of both Theorem~\ref{thm:finite-coarsest-sim}.\ref{thm:finite-coarsest-sim2} and \ref{thm:finite-coarsest-sim}.\ref{thm:finite-coarsest-sim3}.

\begin{proposition}\label{prop:no-coarsest-sim3}
There exists a game $\langle \N, \{0,1\}, W \rangle$, where $W$ is a non-open closed set, and there is a \pO\/ winning strategy, but no maximal harmless constraint.
\begin{proof}
Consider the game $\langle \N, \{0,1\}, W \rangle$, where a run is in $W$ iff it satisfies the following: if \pT\/ always picks $1$, \pO\/ wins; otherwise, the round after the first time that \pT\/ picks $0$, \pO\/ must pick a number greater than the number of $1$ that have been picked by \pT\/ so far. $W$ is closed but not open, the latter can be either checked directly or later inferred from Theorem~\ref{thm:finite-coarsest-sim}.\ref{thm:finite-coarsest-sim3}. Let a constraint $\sim$ be such that there exists a winning $\sim$-strategy $\sO$. Since $\sO$ is winning, $n < \sO(1^{n}0)$ for all $n \in \N$. In particular, $\sO(0) < \sO( 1^{\sO(0)}0)$, so the $\sim$-classes of $(\sO(\e),0)$ and $ \hOS(\sO,1^{\sO(0)}0)$ are distinct since $\sO$ is a $\sim$-strategy. Let $\sim'$ be defined by merging these two classes, so that $\sim'$ is strictly greater than $\sim$. Let us derive a new \pO\/ strategy $\sO'$ as follows: $\sO'(\aH) := \sO(1^{\sO(0)}0)$ if $\aH \sim (\sO(\e),0)$, and $\sO'(\aH) := \sO(\aH)$ otherwise. Thus, $\sO'$ is a winning $\sim'$-strategy.
\end{proof}
\end{proposition}

\bibliography{article}

\end{document}